\documentclass[runningheads]{llncs}

\usepackage{amsfonts,amsmath, amssymb}
\usepackage{empheq}
\usepackage{mathpartir}
\usepackage{semantic}
\usepackage{diagbox}
\usepackage{listings}
\usepackage{bbding}

\newcommand{\drule}[3] {\inferrule*[left={#1}]{#2}{#3}}
\newcommand{\sysstep}{\Rightarrow}
\newcommand{\tsym}[1]{\:\mbox{\texttt{#1}}\:}
\newcommand{\ntsym}[1]{\:\langle\mbox{\emph{#1}}\rangle\:}

\newcommand*\widefbox[1]{\fbox{\hspace{1em}#1\hspace{1em}}}

\newenvironment{bnf}{%
    \scriptsize
    \setkeys{EmphEqEnv}{align*}%
    \setkeys{EmphEqOpt}{box=\widefbox}%
    \EmphEqMainEnv%
}{
    \endEmphEqMainEnv
}

\usepackage[T1]{fontenc}
\usepackage[hyphens]{url}
\begin{document}

\title{From Monolithic to Compositional: A Compositional Operational Semantics for Crystality}

\titlerunning{A Compositional Operational Semantics for Crystality}

\author{Ziyun Xu\inst{1}\orcidID{0009-0003-7824-6502} \and
Hao Wang\inst{2}\orcidID{0000-0003-3557-6301} \and
Meng Sun\inst{1}\orcidID{0000-0001-6550-7396}\inst{ (}\Envelope\inst{)}}

\authorrunning{Z. Xu et al.}

\institute{School of Mathematical Sciences, Peking University, China \and
International Digital Economy Academy (IDEA), China\\
\email{xuziyun@pku.edu.cn, wanghao2020@idea.edu.cn, sunm@pku.edu.cn}}

\maketitle

\begin{abstract}
Parallel execution has become a key approach to improving blockchain scalability, but the lack of formal semantics for smart contract languages in such settings makes rigorous reasoning difficult. Crystality is a smart contract language designed for parallel EVMs, supporting scoped state and asynchronous relay across execution engines. This paper introduces a compositional operational semantics for Crystality. Unlike the original monolithic semantics, the new semantics decomposes the system into engine components and a global component, making the structure of parallel execution explicit. The compositional formulation enables simple proofs of key structural properties, including locality, global isolation, and strong commutativity of independent local steps. Furthermore, we prove that the compositional semantics is semantically equivalent to the original one via a transaction-level bisimulation theorem based on encoding and decoding functions between configurations, and two code-level bisimulation theorems for local and global execution.

\keywords{Operational Semantics \and Compositional Semantics \and Blockchain \and Smart Contracts \and Parallel Execution \and Bisimulation.}

\end{abstract}

\section{Introduction}

Smart contracts enable decentralized applications to operate autonomously on blockchain platforms. As the demand for scalability intensifies \cite{sanka2021systematic}, parallel execution environments, such as parallel Ethereum Virtual Machines (EVMs) \cite{Solana,Aptos,Sui}, have emerged as a promising solution to overcome the inherent throughput limitations of sequential smart contract execution \cite{anjana2019efficient,chen2021forerunner,garamvolgyi2022utilizing,gelashvili2023block,qi2023smart}. However, the lack of formal semantics \cite{HH,WG} for smart contract languages in these parallel settings poses new challenges, particularly in capturing cross-engine interactions and reasoning about distributed state updates.

Crystality \cite{Crystality2025} is a smart contract programming language designed for parallel EVMs. It introduces two key abstractions: {\em Programmable Contract Scope}, which enables fine-grained state partitioning, and {\em Asynchronous Functional Relay}, which supports deferred cross-engine interactions, thereby enabling parallel execution and decentralized control flow.

In \cite{xu2025operational}, we presented an operational semantics \cite{Plotkin} for Crystality that formally defined its behavior and enabled the verification of basic correctness properties via Rocq/Coq \cite{huet1997coq}. While this semantics establishes the first formal foundation for Crystality, it has a key limitation:  the semantics is \emph{monolithic}: the entire system state, including engine states, mempools, and the global state, is encoded in a single configuration. As a result, local and global effects are tightly intertwined in the semantic
rules. The behavior of an individual engine cannot be described independently of the global system configuration. This design makes it difficult to reason modularly about the behavior of individual engines or to establish system-level properties that depend on the separation between local and global computations.

To address this limitation, this paper introduces a \emph{compositional
operational semantics} \cite{abraham2003compositional,cheng1992compositional,mosses2004modular} for Crystality. In the new semantics, the system configuration is decomposed into a collection of \emph{local components}, each representing an engine, together with a distinct \emph{global component} responsible for global transactions. Relay transactions are represented explicitly and propagated through the mempool, while system-level execution composes the behavior of the components. This design exposes the modular structure of Crystality programs and enables local reasoning about engine behavior.

A crucial requirement for such a redesign is that it preserves the meaning of programs. To establish this formally, we define encoding and decoding functions that translate configurations between the original and the compositional semantics. Using these mappings, we prove a \emph{transaction-level bisimulation} theorem which shows that the two semantics are equivalent.

In addition to the equivalence result, we prove several structural properties of the compositional semantics. These include locality of local steps, isolation of global execution, and strong commutativity of independent local steps. These properties are difficult to express directly in the original monolithic semantics but follow naturally from the compositional structure.

The contributions of this paper can be summarized as follows:
\begin{itemize}
\item We introduce a compositional operational semantics for Crystality.
\item We prove a transaction-level bisimulation theorem establishing semantic equivalence between the original non-compositional operational semantics and the proposed compositional operational semantics.
\item We show that the compositional semantics enables clear statements and
proofs of key properties such as locality, global isolation, and strong commutativity of independent local steps.
\end{itemize}
The rest of the paper is organized as follows.
Section~\ref{sec:example} introduces the core features of Crystality. 
Section~\ref{sec:semantics} presents the compositional operational semantics and demonstrates several properties that are 
transparent in the compositional framework.
Section~\ref{sec:correspondence} establishes its correspondence with the
original semantics via bisimulation.
Section~\ref{sec:related} discusses related work, and
Section~\ref{sec:conclusion} concludes.

\section{Crystality}
\label{sec:example}
Crystality is a programming language designed for parallel execution of smart contracts across multiple engines. Instead of assuming a sequential execution context, the language distinguishes different scopes of computation and provides two mechanisms that structure interactions in a concurrent environment. 

\paragraph{Scoped state.}
Crystality programs manipulate state variables that belong to different
\emph{scopes}. The language distinguishes three scopes:
\begin{itemize}
\item \textbf{Address scope}.
State associated with a specific user address. Such variables can be read and updated only by functions executed for that address.
\item \textbf{Engine scope}.
State locally shared in an engine. These variables are visible only within the corresponding engine.
\item \textbf{Global scope}.
State shared by all engines. Global variables represent system-wide information and can be modified only by global functions.
\end{itemize}
The scope discipline restricts how state can be accessed. Local computations may read global variables but cannot modify them directly.

\paragraph{Relay mechanism.}
Crystality provides a primitive \( \texttt{relay@target f(args)} \)
that asynchronously schedules a function invocation at another scope.
Depending on the target, the relay may generate a transaction to be executed by another address, another engine, or at the global level.
The relay mechanism allows a local computation to trigger subsequent actions without blocking the current execution. Relay-generated transactions are placed into mempools associated with the engines and then be selected and executed in subsequent blocks.

Execution in Crystality proceeds in blocks, where each block contains a
sequence of transactions. Since Crystality does not support user-initiated \tsym{global} transactions, any \tsym{global} transaction must be relayed from a previous transaction. Relay transactions needed to be executed first. Therefore, a key convention of Crystality is that \tsym{global} transactions are executed before any \tsym{address} or \tsym{engine} transactions in the same block. The execution of a block therefore follows two phases.
\begin{itemize}
\item \textbf{Global phase}. The global component executes global transactions. The global state may change and the updated global state is propagated to all engines. The mempools may grow due to newly generated relay transactions.
\item \textbf{Local phase}. Once the global transactions of the block have been completed, the global component remains unchanged, while the engines execute their local transactions independently in an interleaving manner.
\end{itemize}

\begin{figure}[t]
\centering
\begin{minipage}{0.9\linewidth}
\begin{verbatim}
contract Bank {
  int @address balance;
  int @global total;
  function transfer(address payee, int amount) @address {
    if (amount <= balance) then {
      balance := balance - amount;
      relay @ payee deposit(amount);
      relay @ global updateTotal(amount);
    } else { skip }
  }
  function deposit(int amount) @address {
    balance := balance + amount;
  }
  function updateTotal(int amount) @global {
    total := total + amount;
  }
}
\end{verbatim}
\end{minipage}
\caption{A simplified banking contract}
\label{fig:bank}
\end{figure}

Figure~\ref{fig:bank} shows a simplified banking contract that maintains account balances and records a global aggregate of transferred funds. Each user account has an instance of variable \texttt{balance} recording the funds associated with the user's address. The global variable \texttt{total} aggregates the amounts transferred across the entire system. The function \texttt{transfer} is executed in the \texttt{address} scope. It deducts tokens from the sender's \texttt{balance} and relays two operations: a \texttt{deposit} to the recipient and an update to the global variable \texttt{total}.

Assume there are two execution engines. Let \texttt{a,d} be addresses managed by engine \texttt{1}, and \texttt{b,c} be addresses managed by engine \texttt{2}. Suppose the current block contains exactly two user-initiated transactions:
\(\texttt{transfer(b,3)} \text{ sent by } \texttt{a}\) and
\(\texttt{transfer(d,2)} \text{ sent by } \texttt{c}\).
When they execute, each \texttt{transfer} deducts tokens from
the sender and produces two relay calls. As a result, relay transactions for \texttt{deposit} and \texttt{updateTotal} are generated and placed into the mempools of the corresponding engines.

In the next block, these relay transactions are executed in two phases. First, the global transactions are executed which update the global variable \texttt{total}. After all global transactions finish, the local transactions are executed which update the balances of the corresponding users.

This example highlights two central aspects of Crystality's execution model: the separation between local computation and global state updates, 
and the deferred execution of relay-generated transactions. These characteristics motivate the design of compositional operational semantics presented in the next section.

\section{Compositional Operational Semantics}
\label{sec:semantics}

We present the compositional operational semantics of Crystality at two levels. The first level presents the semantics of components. Since the execution of transactions within each engine is independent of other engines, each engine can be treated as a relatively autonomous component. Additionally, we treat the global state as a distinct component and refer to it as the 'global component'. Global state variables and global functions are exclusively associated with this global component. The second level encompasses the semantics of the entire system. We have defined compositional rules that integrate the semantics of individual components, thereby deriving the overall system semantics.

In the following, we first introduce the notations used throughout the semantics. We then present the rules for statement execution, expression evaluation, and transaction execution. Finally, we describe how the semantic rules of components are composed to yield the overall semantics of the entire system.

\subsection{Notations}

Let $n$ and $k$ denote the number of execution engines and the number of addresses managed by each engine. The set $\mathbb{A}$ represents all possible memory addresses, while $\mathbb{B}$ denotes the set of byte values. A mapping of the form $\mathbb{A} -> \mathbb{B}$ is used to model storage, associating each memory address with its stored byte contents. $\mathbb{ID}$ refers to the set of all variable identifiers and $\mathbb{IDF}$ denotes the set of function identifiers. Additionally, $\mathbb{T}$ stands for the collection of all variable types.

We define the overall configuration of the system as
$(\Gamma,\Omega,\gamma_G)$, where $\Gamma = (\gamma_1,...,\gamma_n)$. Each $\gamma_i = (\sigma_i,Prog_i)$ represents the configuration of the $i$-th engine, and 
$\gamma_G = (\sigma_G,Prog_G)$ denotes the configuration of the global component. 

$Prog_i$ and $Prog_G$ are the sets of program to be executed. The state of the $i$-th engine is given by $\sigma_i = (\Psi_i,M_i,G_i)$, where $\Psi_i$ captures the storage of state variables (including \tsym{engine} variables and \tsym{address} variables) , $M_i$ is the memory stack used for function calls and temporary variables, and $G_i$ corresponds to the view of the storage of \tsym{global} state variables $G$ from the $i$-th engine's perspective. We introduce $G_i$ To allow each engine to access the global state variables while preserving the relative independence of engines executing in parallel. The state of the global component is defined as $\sigma_G = (G,M_G)$, where $G : \mathbb{A} -> \mathbb{B}$ is the storage of \tsym{global} state variables, and $M_G$ is the memory stack used by the global component. $\Psi_i$ within each engine further decomposes into $\Psi_i = (\Psi_{i,1},\Psi_{i,2},...,\Psi_{i,k},\Psi_{i,s})$, where $\Psi_{i,j}: \mathbb{A} -> \mathbb{B}\, (j=1,...,k)$ holds \tsym{address} state variables specific to the $j$-th address within the $i$-th engine, and $\Psi_{i,s}: \mathbb{A} -> \mathbb{B}$ represents \tsym{engine} state variables within the $i$-th engine.

The memory stack $M_i$ maintains temporary variables of the $i$-th engine, which contains a stack to model new scopes when calling a function. $M_i$ is a stack of functions, each modeled as a mapping $\mathbb{A} -> \mathbb{B}$. The top layer of the stack indicates the memory available in the current function. The stack provides operations such as $top(M_i)$ to access the top layer, $removetop(M_i)$ to remove it, and $addtop(M_i,Y)$ to push a new layer $Y$ to $M_i$. Every layer carries associated contextual information, including a scope attribute denoted by $.scope$ and a return variable identifier denoted by $.rt$, both of which must be defined when a layer is added.

We define the memory pool (mempool) configuration $\Omega = (\Omega_1,...,\Omega_n,\Omega_G)$, where each mempool $\Omega_i$ is associated with the $i$-th engine and contains pending relay transactions, while $\Omega_G$ is the virtual mempool handling global relay transactions.

Finally, we use $\mathcal{E}[\![(\sigma_{\iota},exp_{\iota})]\!]$ for the evaluation of expressions within a given state, where the subscript $\iota$ can refer to any engine index from $1$ to $n$ or to the global component $G$.

For each memory or storage mapping $f: \mathbb{A} -> \mathbb{B}$, we define two auxiliary functions: the name space $N_f: \mathbb{ID} -> \mathbb{A} \cup \mathit{None}$, which maps variable identifiers to their corresponding memory addresses (or indicates that an identifier is undefined), and the type space $T_f: \mathbb{ID} -> \mathbb{T} \cup \mathit{None}$ , which associates each identifier with its declared type. To retrieve the value stored at a particular position, we use the notation $[addr]_{store}^{size}$, which denotes the value of length $size$ located at position $addr$ within the storage $store$. The size of variables of type $t$ is given by $size(t)$.

Additionally, for function identifiers, we introduce several auxiliary functions that retrieve relevant information about the function.
\begin{itemize}
\item $\Lambda_{scope} : \mathbb{IDF} -> \{global,engine,address\}$ determines the declared scope of a function. 
\item $\Lambda_{paraname} : \mathbb{IDF} -> {\mathbb{ID}}^{*}$ returns the ordered list of parameter names associated with each function identifier. 
\item $\Lambda_{paratype} : \mathbb{IDF} -> {\mathbb{T}}^{*}$ specifies the types of those parameters. 
\item $\Lambda_{body} : \mathbb{IDF} -> Prog$ maps each function identifier to its program body.
\item $\Lambda_{rttype} : \mathbb{IDF} -> \mathbb{T} \cup \mathit{None}$  indicates the return type of the function, if any.
\end{itemize}

Given a state $\sigma$, the notation $\sigma'$ represents the state after applying a particular transformation to $\sigma$. In addition, symbols such as $\hat{\sigma}$, $\bar{\sigma}$, or indexed forms like  $\sigma^{(1)}$ and $\sigma^{(2)}$ denote intermediate states arising during a sequence of modifications. The name space and type space associated with a function are auxiliary but integral components of its definition. Consequently, any change to the name space or type space also implies a change in the function itself.

Furthermore, there is an implicit 'zooming' mechanism in the notation. For example, if a component such as $\Psi_{i,s}$ is updated to $\Psi_{i,s}'$, then the enclosing state $\sigma_i$ is implicitly updated to $\sigma_i'$, while all other components remain unchanged. This convention allows the formal definitions to remain concise and more readable without requiring explicit expansion of all nested structures each time a modification is described.

In the semantic rules, the arrows are annotated with superscripts and subscripts. There are two types of subscripts: $l$ and $g$, indicating whether the rule is executed in a single engine or by the global component. The superscripts on the arrows also come in two forms: $\tau$ and $?(\omega_1,...,\omega_n,\omega_G)$ which differentiate whether the corresponding step involves interaction with other components. $\tau$ indicates that the step does not generate any relay calls, while $?(\omega_1,...,\omega_n,\omega_G)$ signifies that the step sends relay transactions to various memory pools.

\subsection{Semantic Rules of Statements}

We now present the first-level rules defining the compositional operational semantics. We begin by describing how state variables are declared during contract deployment.

When a state variable is declared with \tsym{address} scope, the system creates a separate instance of that variable for every valid address managed by each engine. Storage space is allocated for each instance according to the variable's declared type. In rule \textsc{SVDa}, the first premise asserts that, prior to initialization, the \tsym{address} variable named $id$ is undefined in $\sigma_i$. The second premise updates the type space in the resulting state $\sigma_i'$, recording that $id$ has the specified type $Type$. The third step allocates memory for $id$ in the name space of $\sigma_i'$. Finally, the fourth step initializes the allocated memory with the default value corresponding to $Type$. Because there is one such instance per address, these operations must be performed for all $k$ addresses within the engine.
\begin{gather*}
\drule{SVDa}{ N_{\Psi_{i,j}}(id) = \mathit{None}  ,\ 1\le j\le k \\ T_{\Psi_{i,j}'}(id) = Type ,\ 1\le j\le k  \\ N_{\Psi_{i,j}'}(id) = allocate\_new(Type,\Psi_{i,j}) ,\ 1\le j\le k\\ [N_{\Psi_{i,j}'}(id)]_{\Psi_{i,j}'}^{size(Type)} = init(Type) ,\ 1\le j\le k  }{ (\sigma_i, @address\ Type\ id;) ->_l^{\tau} (\sigma_i', \cdot) }
\end{gather*}
The rule \textsc{SVDg} requires special attention. In the compositional semantics, we conceptualize accesses to \tsym{global} state variables and invocations of \tsym{global} functions as operating within an independent, virtual component separate from any specific engine. This abstraction is a modeling convenience for clarity and modularity, and does not necessarily correspond to the system's implementation.
\begin{gather*}
\drule{SVDg}{N_G(id) = \mathit{None}\\ T_{G'}(id) = Type \\\\ N_{G'}(id)=allocate\_new(Type,G)\\ [N_{G'}(id)_{G'}^{size(Type)}] = init(Type)  }{(\sigma_G,@global\ Type\ id;) ->_g^{\tau} (\sigma_G',\cdot)}
\end{gather*}
The rules for declaring temporary variables require a distinction based on the location where they occur. A statement is local if it is within an \tsym{address} or \tsym{engine} function, affecting the memory stack of the corresponding engine. Conversely, if the statement is within a \tsym{global} function, it should affect the memory of the global component.
\begin{gather*}
\drule{TVDl}{top(M_i).scope \ne global \\ N_{top(M_i)}(id) = \mathit{None} \\ T_{top(M_i)'}(id) = Type \\ N_{top(M_i)'}(id) = allocate\_new(Type,top(M_i)) \\ [N_{top(M_i)'}(id)]_{top(M_i)'}^{size(Type)} = init(Type)}{(\sigma_i, Type\ id;) ->^{\tau}_l (\sigma_i',\cdot)}
\\[0.5em]
\drule{TVDg}{top(M_G).scope = global \\ N_{top(M_G)}(id) = \mathit{None} \\ T_{top(M_G)'}(id) = Type \\ N_{top(M_G)'}(id) = allocate\_new(Type,top(M_G)) \\ [N_{top(M_G)'}(id)]_{top(M_G)'}^{size(Type)} = init(Type)}{(\sigma_G, Type\ id;) ->^{\tau}_g (\sigma_G',\cdot)}
\end{gather*}
We illustrate assignment statements using the rule \textsc{SASSIGNaa}. The rule checks the scope of the function performing the assignment and the scope of the state variable, and applies when both scopes are \tsym{address}. The expression $exp$ here may be a function call with a return value. The states $\sigma_i$ and $\hat{\sigma_i}$ denote the configuration of the current engine before and after the evaluation of $exp$, respectively.

\begin{gather*}
\drule{SASSIGNaa}{\mathcal{E}[\![(\sigma_i,  exp)]\!] ->_l^{?(\omega_1,...,\omega_n,\omega_G)} (\hat{\sigma_i},v)   \\\\ top(M_i).scope = (address,j) \\ top(\hat{M_i}).scope = (address,j) \\\\ N_{\hat{\Psi_{i,j}}}(id) \ne \mathit{None} \\ [N_{\hat{\Psi_{i,j}}'}(id)]^{size(T_{\hat{\Psi_{i,j}}'}(id))}_{\hat{\Psi_{i,j}}'}=v }{(\sigma_i,  id = exp;) ->_l^{?(\omega_1,...,\omega_n,\omega_G)} (\hat{\sigma_i}',\cdot)}
\end{gather*}

We use rule \textsc{IFUNgg} as an example for function calls. The key point is to 'record' any relay calls generated during the process, and to ensure that all these steps are marked as global. By treating the global environment as a separate component, we can more easily define the rules for \tsym{global} function calls than in the original semantics. 
All state changes will occur within the global component because \tsym{global} function is not allowed to modify \tsym{engine} or \tsym{address} variables.
\begin{gather*}
\drule{IFUNgg}{\Lambda_{scope}(idf)=global \\ top(M_G).scope = global\\\\
\Lambda_{rettype}(idf) = \mathit{None} \\ \Lambda_{body}(idf) = Block \\\\
\Lambda_{paraname}(idf) =(id_1,...,id_t) \\ \Lambda_{paratype}(idf) = (T_1,...,T_t)\\\\
M_G^{(1)}=addtop(M_G,top(M_G))\\
top(M_G^{(1)}).scope = global \\\\
top(M_G^{(1)}).rt = \mathit{None} \\
M_G^{(3)} = removetop(M_G^{(2)}) \\\\
(\sigma_G^{(1)},T_1\ id_1 = exp_1;...;T_t\ id_t = exp_t;Block) ->_g^{?(\omega_1,...,\omega_n,\omega_G)} (\sigma_G^{(2)},\cdot) }{(\sigma_G,idf(exp_1,...,exp_t)) ->_g^{?(\omega_1,...,\omega_n,\omega_G)} (\sigma_G^{(3)},\cdot)}
\end{gather*}

We separate all memory pools from the state of each engine, treating them as distinct components. Since transactions in the memory pool do not affect the execution of the current transaction, this separation allows for a clearer representation of Crystality's semantics.

The engine or global component generating the relay transactions 'sends' them out. In practice, the relay transactions are propagated through the network, selected by nodes, and stored in memory pools, which is a process that may involve some delay. Therefore, we separate the sending and receiving of relay transactions into different rules.

We use the rule \textsc{RELAYal} as an example for relay calls. It handles the case where the relay target is a specific address. Evaluating the expression $exp$ yields a result $(r, j)$, which identifies the $j$-th address managed by the $r$-th engine. The function $idf$, along with its provided arguments, is then encapsulated into a relay transaction and dispatched to the engine responsible for the target address.
\begin{gather*}
\drule{RELAYal}{top(M_i).scope \ne global \\ \Lambda_{scope}(idf) = address \\\\ \mathcal{E}[\![(\sigma_i,exp)]\!] ->_l^{\tau} (\sigma_i,(r,j)) \\  \mathcal{E}[\![(\sigma_i,exp_l)]\!] ->_l^{\tau} (\sigma_i,v_l), 1 \le l \le t  \\\\
\omega_r= \{ (address,j,idf(v_1,...,v_t)) \} \\
\omega_\iota = \emptyset , \forall \iota \ne r }{(\sigma_i,relay@exp\ idf(exp_1,...,exp_t);) ->_l^{?(\omega_1,...,\omega_n,\omega_G)} (\sigma_i,\cdot)}
\end{gather*}

The sending of relay transactions is denoted by $?(\omega_1,...,\omega_n,\omega_G)$, while the receiving is denoted by $!(\omega_1,...,\omega_n,\omega_G)$.

\begin{gather*}
\drule{RELAYo}{\quad}{(\Omega_1,...,\Omega_n,\Omega_G) ->^{!(\omega_1,...,\omega_n,\omega_G)} (\Omega_1 \cup \omega_1,...,\Omega_n \cup \omega_n,\Omega_G \cup \omega_G)}
\end{gather*}

Due to space limitation, the rules for other statements are omitted here and provided in the appendix.

\subsection{Semantic Rules of Evaluations} 

The rules for evaluating expressions are also provided compositionally.
Due to space limitations, we do not present all the evaluation rules here. Instead, we focus on two representative ones, \textsc{ESIDag} and \textsc{ESIDgg}. In Crystality, the \tsym{address} and \tsym{engine} function can read \tsym{global} state variables. Therefore, when modeling the behavior of each engine, it is necessary to bind a space corresponding to all \tsym{global} variables within the state of each engine. However, we model \tsym{global} variables and \tsym{global} functions within a separate, virtual component. As a result, the corresponding entity within each engine's state is essentially what that engine perceives as '$G$', namely $G_i$ for the $i$-th engine. In the first level of the compositional semantics, we do not guarantee that the perceived $G_i$ across different engines, as well as the actual $G$, are identical. However, in the second level of the semantics, which concerns the overall system semantics, we define a corresponding composition operation that ensures they are identical.
\begin{gather*}
\drule{ESIDag}{top(M_i).scope = (address,j) \\ N_{G_i}(id) \ne \mathit{None} \\ [N_{G_i}(id)]^{size(T_{G_i}(id))}_{G_i}=v }{\mathcal{E}[\![(\sigma_i,id) ]\!]->_l^{\tau}(\sigma_i,v)}
\end{gather*}
\textsc{ESIDgg} corresponds to the case where a \tsym{global} function reads a \tsym{global} state variable. This can be directly accessed from $G$. In contrast to \textsc{ESIDag}, there is no need to go through the $G_i$ perceived by some engine.
\begin{gather*}
\drule{ESIDgg}{top(M_G).scope = global \\ N_G(id) \ne \mathit{None}\\
[N_G(id)]^{size(T_G(id))}_G=v }{\mathcal{E}[\![(\sigma_G,id) ]\!]->_g^{\tau}(\sigma_G,v)}
\end{gather*}

\subsection{Semantic Rules of Transactions}

Defining transaction semantics is central to smart contract languages. We model transactions as function calls (T-functions), abstracting away gas and signatures.

Since we first present the semantics at the component level, the rules for a transaction only involve the engine where the transaction occurs (or the global component if it is a \tsym{global} transaction). We introduce a subscript $t$ on the arrow to indicate that this step represents the complete execution of a transaction, specifically a T-function.

Unlike the operational semantics in \cite{xu2025operational}, the compositional semantics allows for a more precise characterization of transaction properties. This is achieved through system-level semantics, which is presented in the next subsection. 

We focus on the atomicity and inherent parallelism of transactions. A transaction is an atomic operation. That is, all (complete done) or nothing (zero effect). Therefore, we explicitly distinguish whether a computation is a complete transaction by the subscript $t$. On the other hand, transactions in different engines are executed in parallel. Each engine operates independently with its own state. The only interaction between engines is through message passing (relay transactions), which does not affect the execution of the current transaction. This parallelism is captured in the system-level semantics.

The rules for transactions are provided in the appendix.

\subsection{Semantic Rules of the System} 

We first observe that any rule annotated with the superscript $\tau$
can be regarded as a rule annotated with $?(\omega_1,...,\omega_n,\omega_G)$, by taking each $\omega_\iota$ to be the empty set.

We define a function $\mu(\Gamma,\Omega,\gamma_G) =_{df} G_1=G_2=...=G_n=G$ that ensures the corresponding components $G_i$ and $G$ in $\Gamma$ and $ \gamma_G$ are all equal, and another function $\rho(\Gamma_1,\Gamma_2)$ that ensures $\Gamma_1$ and $\Gamma_2$ are identical except for the components $G_i$ in each of them.

Due to the atomicity of transactions, we only permit the configuration before and after a T-function (a transaction) to constitute the system's configuration. This ensures that at the system level all transactions maintain atomicity.

The execution of a \tsym{global} T-function will only modify $\gamma_G$ with the resulting relay calls reflected in the update of $\Omega$. Before and after execution, the state of each engine remains unchanged, except that the $G_i$ perceived by each engine must be consistent with $G$.
\begin{gather*}
\drule{G}{\gamma_G ->_{g,t}^{?(\omega_1,...,\omega_n,\omega_G)} \gamma_G' \\  \Omega ->^{!(\omega_1,...,\omega_n,\omega_G)} \Omega' \\\\ \rho(\Gamma,\Gamma') \\ \mu(\Gamma,\Omega,\gamma_G) \\ \mu(\Gamma',\Omega',\gamma_G')}{(\Gamma,\Omega,\gamma_G) \sysstep_g (\Gamma',\Omega',\gamma_G') }
\end{gather*}
The execution of an \tsym{address} or \tsym{engine} T-function will only modify the configuration of the corresponding engine. The configurations of other engines and the configuration $\gamma_G$ remain unchanged.
According to the execution model, such an execution can only occur at the local phase, that is after all global transactions have been completed. Therefore, we require that $Prog_G = \cdot$.
\begin{gather*}
\drule{L}{\gamma_i ->_{l,t}^{?(\omega_1,...,\omega_n,\omega_G)} \gamma_i' \\  \Omega ->^{!(\omega_1,...,\omega_n,\omega_G)} \Omega'\\
\Gamma=(\gamma_1,...,\gamma_i,...,\gamma_n)\\ \Gamma'=(\gamma_1,...,\gamma_i',...,\gamma_n) \\ \gamma_G = (\sigma_G,\cdot) \\ \mu(\Gamma,\Omega,\gamma_G) \\ \mu(\Gamma',\Omega',\gamma_G)}{(\Gamma,\Omega,\gamma_G) \sysstep_l (\Gamma',\Omega',\gamma_G)}
\end{gather*}
The rules \textsc{Ga}, \textsc{Gs}, \textsc{La}, \textsc{Ls}, and \textsc{GL} describe the changes before and after several transactions. We express by $ \sysstep_l^{*} $ that there is a chain of local steps, by $ \sysstep_g^{*} $ that there is a chain of global steps, and by $ \sysstep_{gl}^{*} $ that there is a chain of global steps followed by local steps (See the execution convention). We use \textsc{Ga}, \textsc{Gs} and \textsc{GL} as examples.
\begin{gather*}
\drule{Ga}{ \mu(\Gamma,\Omega,\gamma_G)  }{(\Gamma,\Omega,\gamma_G) \sysstep_g^{*} (\Gamma,\Omega,\gamma_G) }
\\[0.5em]
\drule{Gs}{(\Gamma,\Omega,\gamma_G) \sysstep_g (\Gamma',\Omega',\gamma_G')  \\ (\Gamma',\Omega',\gamma_G') \sysstep_g^{*} (\Gamma'',\Omega'',\gamma_G'') }{(\Gamma,\Omega,\gamma_G) \sysstep_g^{*} (\Gamma'',\Omega'',\gamma_G'')}
\\[0.5em]
\drule{GL}{(\Gamma,\Omega,\gamma_G) \sysstep_g^{*} (\Gamma',\Omega',\gamma_G')  \\ (\Gamma',\Omega',\gamma_G') \sysstep_l^{*} (\Gamma'',\Omega'',\gamma_G') }{(\Gamma,\Omega,\gamma_G) \sysstep_{gl}^{*} (\Gamma'',\Omega'',\gamma_G')}
\end{gather*}

We present several properties that can be naturally stated and easily proved under the compositional semantics. The proofs are provided in the appendix.

\begin{theorem}[Locality]
\label{thm:locality}
For any system-level local step
\(
(\Gamma,\Omega,\gamma_G)\sysstep_l(\Gamma',\Omega',\gamma_G'),
\)
there exists a unique \(i \in \{1,\ldots,n\}\) such that \(\Gamma\) and
\(\Gamma'\) differ only at the \(i\)-th component. Moreover, \( \gamma_G'=\gamma_G, \) and \(\Omega'\) is obtained from \(\Omega\) by union with the relay sets produced in the step.
\end{theorem}

\begin{theorem}[Global isolation]
\label{thm:global-isolation}
For any system-level global step
\(
(\Gamma,\Omega,\gamma_G)\sysstep_g(\Gamma',\Omega',\gamma_G'),
\)
the engine components may change only in their copies of the global
state. And \(\Omega'\) is obtained from \(\Omega\) by union with the relay
sets produced in the step.
\end{theorem}

\begin{theorem}[Strong commutativity of independent local steps]
\label{thm:strong-comm-local}
Suppose
\(
C \sysstep_l C_i
\) and \(
C \sysstep_l C_j
\)
are two system-level local steps such that the first step updates
engine \(i\) and the second step updates engine \(j\), where \(i\neq j\).
Then there exists a unique configuration \(C'\) such that
\(
C_i \sysstep_l C'
\) and \(
C_j \sysstep_l C'.
\)
Equivalently, the two independent local steps commute.
\end{theorem}

Theorem~\ref{thm:strong-comm-local} is the well-known diamond (commutativity) property. It can be proved easily using the compositional structure together with Theorem~\ref{thm:locality}.

\section{Correspondence with the Original Semantics}
\label{sec:correspondence}

The compositional semantics introduced in the previous section restructures
the execution model of Crystality by separating local and global component. A natural question is whether this new semantics faithfully represents the behavior of the original one. The main result of this section is a \emph{bisimulation theorem} at transaction-level, which proves that the compositional semantics is \emph{semantically equivalent} to the previous one \cite{xu2025operational}. The proof is organized around two ingredients. First, we define \emph{encoding} and \emph{decoding} functions between the configurations of the two semantics and show that they are mutual inverse. Second, to lift the correspondence from transaction-level to the interior of a transaction, we introduce a code-level correspondence and prove bisimulation theorems for local and global code-level execution. For brevity, we present only the essential definitions and lemmas. Detailed definitions and full proofs are given in the appendix.

Throughout this section, symbols from the original
(non-compositional) semantics are annotated with the superscript
\(\mathsf{o}\), while unannotated symbols refer to the compositional semantics. In particular, we write \(C^{\mathsf{o}}\), \(\sigma_i^{\mathsf{o}}\), \(\Omega_i^{\mathsf{o}}\),
\(Prog_i^{\mathsf{o}}\), \(G^{\mathsf{o}}\) and \(\to^{\mathsf{o}}\) for the notions in the original semantics.

In the following, we use several auxiliary functions to manipulate transaction sequences and mempools. The functions \(\mathsf{gprefix}\) and \(\mathsf{lsuffix}\) extract the global prefix and the local suffix of a program, respectively, reflecting the convention that global transactions
precede local ones in Crystality. The functions \(\mathsf{LTx}\) and
\(\mathsf{GTx}\) extract from a mempool the local and global
transactions. Their formal definitions are given in the appendix.

We first define well-formed configurations for both the original and the compositional semantics, since the correspondence between the two semantics is established only on well-formed configurations.

\begin{definition}[Well-formed configurations]
Let
\[
\Gamma = (\gamma_1,\ldots,\gamma_n), \qquad
\gamma_i = (\sigma_i, Prog_i), \qquad
\sigma_i = (\Psi_i,M_i,G_i),
\]
\[
\gamma_G = (\sigma_G, Prog_G), \qquad
\sigma_G = (G,M_G).
\]
We define the well-formedness predicate for configurations of the
compositional semantics by
\[
\mathsf{wf}((\Gamma,\Omega,\gamma_G))
\;\stackrel{\mathrm{def}}{=}\;
G_1 = G_2 = \cdots = G_n = G.
\]

For a configuration of the original semantics
\[
C^{\mathsf{o}}
=
(\sigma_1^{\mathsf{o}},\Omega_1^{\mathsf{o}},Prog_1^{\mathsf{o}},\ldots,
\sigma_n^{\mathsf{o}},\Omega_n^{\mathsf{o}},Prog_n^{\mathsf{o}},G^{\mathsf{o}}),
\]
we define
\[
\mathsf{wf}^{\mathsf{o}}(C^{\mathsf{o}})
\;\stackrel{\mathrm{def}}{=}\;
\mathsf{gprefix}(Prog_1^{\mathsf{o}})
=
\mathsf{gprefix}(Prog_2^{\mathsf{o}})
=
\cdots
=
\mathsf{gprefix}(Prog_n^{\mathsf{o}}).
\]
\end{definition}

We then present the definitions of the two key functions, \(Enc\) and \(Dec\) .

\begin{definition}[Encoding] \label{def:enc}
Let \( C^{\mathsf{o}} =
(\sigma_1^{\mathsf{o}},\Omega_1^{\mathsf{o}},Prog_1^{\mathsf{o}},\ldots,
\sigma_n^{\mathsf{o}},\Omega_n^{\mathsf{o}},Prog_n^{\mathsf{o}},G^{\mathsf{o}}) \)
be a well-formed configuration of the original semantics,
where \(\sigma_i^{\mathsf{o}}=(\Psi_i^{\mathsf{o}},M_i^{\mathsf{o}})\).

Let
\[
Prog_G = \mathsf{gprefix}(Prog_i^{\mathsf{o}}),
\qquad
Prog_i = \mathsf{lsuffix}(Prog_i^{\mathsf{o}}),
\]
\[
\Omega_i = \mathsf{LTx}(\Omega_i^{\mathsf{o}}),
\qquad
\Omega_G = \mathsf{GTx}(\Omega_i^{\mathsf{o}}).
\]

Define
\[
M_i =
\begin{cases}
\emptyset, & Prog_G \neq \cdot,\\
M_i^{\mathsf{o}}, & Prog_G = \cdot,
\end{cases}
\qquad
M_G =
\begin{cases}
M_i^{\mathsf{o}}, & Prog_G \neq \cdot,\\
\emptyset, & Prog_G = \cdot.
\end{cases}
\]

The encoding function
\(
Enc : Conf^{\mathsf{o}} \rightarrow Conf
\)
maps \(C^{\mathsf{o}}\) to the configuration
\((\Gamma,\Omega,\gamma_G)\) defined by
\[
\Gamma = (\gamma_1,\ldots,\gamma_n),
\qquad
\gamma_i = ((\Psi_i^{\mathsf{o}},M_i,G^{\mathsf{o}}),Prog_i),
\]
\[
\Omega = (\Omega_1,\ldots,\Omega_n,\Omega_G),
\qquad
\gamma_G = ((G^{\mathsf{o}},M_G),Prog_G).
\]
\end{definition}

\begin{definition}[Decoding] \label{def:dec}
Let \( (\Gamma,\Omega,\gamma_G) \)
be a well-formed configuration of the compositional semantics, where
\[
\Gamma=(\gamma_1,\ldots,\gamma_n),
\qquad
\gamma_i=((\Psi_i,M_i,G),Prog_i),
\]
\[
\Omega=(\Omega_1,\ldots,\Omega_n,\Omega_G),
\qquad
\gamma_G=((G,M_G),Prog_G).
\]
Define
\[
M_i^{\mathsf{o}} =
\begin{cases}
M_G, & Prog_G \neq \cdot,\\
M_i, & Prog_G = \cdot .
\end{cases}
\]

The decoding function
\(
Dec : Conf \rightarrow Conf^{\mathsf{o}}
\)
maps \((\Gamma,\Omega,\gamma_G)\) to the configuration
\(
(\sigma_1^{\mathsf{o}},\Omega_1^{\mathsf{o}},Prog_1^{\mathsf{o}},\ldots,
\sigma_n^{\mathsf{o}},\Omega_n^{\mathsf{o}},Prog_n^{\mathsf{o}},G^{\mathsf{o}})
\)
defined by
\[
\sigma_i^{\mathsf{o}} = (\Psi_i,M_i^{\mathsf{o}}),
\quad
\Omega_i^{\mathsf{o}} = \Omega_i \cup \Omega_G,
\quad
Prog_i^{\mathsf{o}} = Prog_G ; Prog_i,
\quad
G^{\mathsf{o}} = G .
\]
\end{definition}

The functions \(Enc\) and \(Dec\) are structurally dual.
\(Enc\) decomposes a configuration of the original
semantics into the components of the compositional semantics:
the global prefix of each program is extracted into the global
component, the remaining local suffix stays in each engine,
the global transactions in the mempools are separated into
\(\Omega_G\), and the stacks are distributed between the local
and global components depending on whether a global computation
is in progress. Conversely, \(Dec\) reconstructs a configuration of the original semantics by merging these components:
the global program \(Prog_G\) is prefixed to each local program,
the global mempool \(\Omega_G\) is merged into each engine mempool,
and the stacks are restored according to whether the global
program is empty. \(Enc\) factorizes a monolithic configuration
into compositional components, while \(Dec\) assembles these
components back into a monolithic configuration.

\begin{lemma}[Well-formedness preserved by \(Enc\) and \(Dec\)]
\label{lem:wf-preserve-enc-dec}
If \(\mathsf{wf}^{\mathsf{o}}(C^{\mathsf{o}})\) , then \(\mathsf{wf}(Enc(C^{\mathsf{o}}))\).
If \(\mathsf{wf}((\Gamma,\Omega,\gamma_G))\), then
\(\mathsf{wf}^{\mathsf{o}}(Dec((\Gamma,\Omega,\gamma_G)))\) .
\end{lemma}

\begin{lemma}[Mutual inverse properties of \(Enc\) and \(Dec\)]
\label{thm:enc-dec-inverse}
The functions \(Enc\) and \(Dec\) are mutually inverse on reachable well-formed configurations.

(i) For every reachable well-formed configuration \(C^{\mathsf{o}}\) of the original semantics,
\(
Dec(Enc(C^{\mathsf{o}})) = C^{\mathsf{o}}.
\)

(ii) For every reachable well-formed configuration \(C\) of the compositional semantics,
\(
Enc(Dec(C)) = C.
\)

\end{lemma}

The functions \(Enc\) and \(Dec\) are defined on all well-formed configurations. We are particularly interested in configurations that occur at \emph{transaction boundaries}, i.e., immediately before a transaction begins or after it completes. Transaction-boundary configurations are formally defined in the appendix. Restricting \(Enc\) and \(Dec\) to such configurations yields the functions \(Enc_T\) and \(Dec_T\).

Using \(Enc_T\) and \(Dec_T\), we define a transaction-level correspondence between configurations of the original semantics and the compositional semantics.

\begin{definition}[Transaction-level correspondence]
Let \(Conf^{\mathsf{o}}\) and \(Conf\) denote the sets of configurations
of the original semantics and the compositional semantics,
respectively.

We define a relation
\[
R_T \subseteq Conf^{\mathsf{o}} \times Conf
\]
such that \(R_T(C^{\mathsf{o}},C)\) holds iff

\begin{enumerate}
\item \(C^{\mathsf{o}}\) is reachable and well formed in the original semantics;
\item \(C\) is reachable and well formed in the compositional semantics;
\item both \(C^{\mathsf{o}}\) and \(C\) are at transaction boundaries; and
\item \(Enc_T(C^{\mathsf{o}}) = C\)
      (equivalently, \(Dec_T(C)=C^{\mathsf{o}}\)).
\end{enumerate}
\end{definition}

We now state the main theorem of this section. Due to space limitations, the proof is deferred to the appendix.

\begin{theorem}[Transaction-level bisimulation]
\label{thm:transaction-bisim}
Let \(R_T\) be the transaction-level correspondence relation.
If \(R_T(C^{\mathsf{o}},C)\), then the following properties hold:

(i) If \( C^{\mathsf{o}} \to C^{\mathsf{o}\prime} \)
by a transaction step of the original semantics,
then there exists a configuration \(C'\) such that
\( C \sysstep C' \)
and \(R_T(C^{\mathsf{o}\prime},C')\).

(ii) If \(C \sysstep C'\)
then there exists a configuration \(C^{\mathsf{o}\prime}\) such that
\( C^{\mathsf{o}} \to C^{\mathsf{o}\prime} \)
by a transaction step of the original semantics,
and \(R_T(C^{\mathsf{o}\prime},C')\).
\end{theorem}

Since each transaction is an atomic unit of execution, the correspondence is defined at transaction boundaries, the equivalence
between the original semantics and the compositional semantics is
naturally stated at transaction boundaries, and the main bisimulation
theorem (Theorem~\ref{thm:transaction-bisim}) relates configurations only at these boundaries.

However, proving this result requires reasoning about the execution
inside a transaction. In the compositional semantics, the behavior of the system is defined component-wise: either a local component \(\gamma_i\) executes local code, or the global component \(\gamma_G\) executes global code. To relate these component-level executions with the monolithic
execution of the original semantics, we introduce system configurations at the code level.

These configurations should be understood purely as proof artifacts.
They do not correspond to observable states in the actual execution,
since system-level composition occurs only at transaction boundaries.
Instead, they serve as an intermediate bridge in the proof:
\[
\gamma
\;\xrightarrow{\text{Wrap}}\;
C
\;\xleftrightarrow{\;\; R \;\;}\;
C^{\mathsf{o}},
\]
where \(\gamma\) denotes either a local component \(\gamma_i\) or the
global component \(\gamma_G\), \(C\) is a system configuration of the
compositional semantics, and \(C^{\mathsf{o}}\) is a configuration of
the original semantics. The relation \(R\), defined using the code-level mappings \(Enc\) and \(Dec\), enables us to relate component-level execution with the execution of the entire system in the original semantics.

To establish bisimulation at transaction-level, we first prove two code-level bisimulation theorems, corresponding to local and global execution respectively. In the compositional semantics, code-level rules operate on component configurations \(\gamma\), whereas in the original semantics they operate on whole-system configurations. Therefore, several auxiliary constructions and definitions are introduced to relate component configurations in the compositional semantics with system configurations in the original semantics. We show that one step of a component in the compositional semantics can be simulated by a corresponding system step in the original semantics, and vice versa.

We directly state the two code-level bisimulation theorems below; the necessary definitions, lemmas, and full proofs are provided in the appendix.

\begin{theorem}[Local code-level bisimulation]
\label{lem:local-code-bisim}
If \(R_l(i,\gamma_i,C^{\mathsf{o}})\), then the following properties hold:

(i) If
\(
\gamma_i \to_l^{?(w_1,\ldots,w_n,w_G)} \gamma_i',
\)
then there exists a configuration \(C^{\mathsf{o}\prime}\) such that
\(
C^{\mathsf{o}} \to C^{\mathsf{o}\prime},
\) \(
\mathsf{Same}_l^{\mathsf{o}}(i,C^{\mathsf{o}},C^{\mathsf{o}\prime}),
\) \(
\mathsf{MemDiff}^{\mathsf{o}}(C^{\mathsf{o}},C^{\mathsf{o}\prime},\Delta\Omega),
\)
where
\(
\Delta\Omega=(w_1\cup w_G,\; w_2\cup w_G,\; \ldots,\; w_n\cup w_G),
\)
and
\(
R_l(i,\gamma_i',C^{\mathsf{o}\prime}).
\)

(ii) If
\(
C^{\mathsf{o}} \to C^{\mathsf{o}\prime},
\) \(
\mathsf{Same}_l^{\mathsf{o}}(i,C^{\mathsf{o}},C^{\mathsf{o}\prime}),
\) \(
\mathsf{MemDiff}^{\mathsf{o}}(C^{\mathsf{o}},C^{\mathsf{o}\prime},\Delta\Omega),
\)
where
\(
\Delta\Omega=(\Delta\Omega_1,\ldots,\Delta\Omega_n),
\)
then there exists a component \(\gamma_i'\) such that
\(
\gamma_i \to_l^{?(w_1,\ldots,w_n,w_G)} \gamma_i',
\)
where
\(
w_j = \mathsf{LTx}(\Delta\Omega_j)\quad (1\le j\le n),
\)
and
\(
w_G = \mathsf{GTx}(\Delta\Omega_j)
\)
for any \(j \in \{1,\ldots,n\}\), and
\(
R_l(i,\gamma_i',C^{\mathsf{o}\prime}).
\)
\end{theorem}

\begin{theorem}[Global code-level bisimulation]
\label{lem:global-code-bisim}
If \(R_g(\gamma_G,C^{\mathsf{o}})\), then the following properties hold:

(i) If
\(
\gamma_G \to_g^{?(w_1,\ldots,w_n,w_G)} \gamma_G',
\)
then there exists a configuration \(C^{\mathsf{o}\prime}\) such that
\(
C^{\mathsf{o}} \to C^{\mathsf{o}\prime},
\)
\(
\mathsf{Same}_g^{\mathsf{o}}(C^{\mathsf{o}},C^{\mathsf{o}\prime}),
\)
\(
\mathsf{MemDiff}^{\mathsf{o}}(C^{\mathsf{o}},C^{\mathsf{o}\prime},\Delta\Omega),
\)
where
\(
\Delta\Omega=(w_1\cup w_G,\; w_2\cup w_G,\; \ldots,\; w_n\cup w_G),
\)
and
\(
R_g(\gamma_G',C^{\mathsf{o}\prime}).
\)

(ii) If
\(
C^{\mathsf{o}} \to C^{\mathsf{o}\prime},
\)
\(
\mathsf{Same}_g^{\mathsf{o}}(C^{\mathsf{o}},C^{\mathsf{o}\prime}),
\)
\(
\mathsf{MemDiff}^{\mathsf{o}}(C^{\mathsf{o}},C^{\mathsf{o}\prime},\Delta\Omega),
\)
where
\(
\Delta\Omega=(\Delta\Omega_1,\ldots,\Delta\Omega_n),
\)
then there exists a component \(\gamma_G'\) such that
\(
\gamma_G \to_g^{?(w_1,\ldots,w_n,w_G)} \gamma_G',
\)
where
\(
w_i = \mathsf{LTx}(\Delta\Omega_i) \quad (1\le i\le n),
\)
and
\(
w_G = \mathsf{GTx}(\Delta\Omega_i)
\)
for any \(i \in \{1,\ldots,n\}\), and
\(
R_g(\gamma_G',C^{\mathsf{o}\prime}).
\)
\end{theorem}

\section{Related Work} \label{sec:related}

Formal semantics of smart contract languages has been extensively studied, particularly in the context of sequential execution platforms such as Ethereum.  Bhargavan \cite{bhargavan2016formal} formalized a core subset of Solidity \cite{Solidity} in F*, enabling verification of safety properties.  Grishchenko \cite{grishchenko2018semantic} developed the KEVM semantics using the K-framework, capturing the EVM bytecode precisely and supporting automated reasoning. Marmsoler \cite{Marmsoler} introduced a denotational semantics of Solidity.  Zakrzewski  \cite{zakrzewski2018towards} formalized a big-step semantics for a subset of Solidity. Crosara and Jiao   \cite{crosaratowards,jiao2018,jiao2020semantic} proposed operational semantics, including executable definitions of Solidity behaviors. Other approaches encode Solidity into intermediate representations:  Hajdu \cite{Hajdu1,Hajdu2} developed an SMT-based intermediate language to support verification, and Yang \cite{Yang2020} presented a formal semantics for a specification language targeting Ethereum contracts in Coq. Luu \cite{luu2016making} focused on symbolic analysis of smart contracts to detect common vulnerabilities. These approaches primarily target sequential models and do not address the unique challenges introduced by parallel execution.

Parallel execution in blockchain systems has been proposed to improve scalability. Notable examples include proposals for parallel transaction execution engines \cite{dickerson2017adding}, and experimental systems such as Chainspace \cite{al2017chainspace}. While these systems provide architectural solutions to concurrency, they rarely define formal operational semantics that explicitly capture cross-engine interactions and asynchronous relay mechanisms at the language level.

Structural operational semantics (SOS) has a long tradition in programming language theory \cite{Plotkin}. Modular SOS frameworks, such as MSOS \cite{mosses2004modular}, emphasize compositional definitions to enable scalable reasoning about program fragments and language extensions. Our work draws inspiration from these principles but adapts them to the domain of parallel smart contracts, where partitioned state and deferred cross-engine calls require novel semantic treatment. To our knowledge, this paper is the first to define a compositional operational semantics for a smart contract language designed for parallel execution environments.

\section{Conclusion} \label{sec:conclusion}

This paper introduces a compositional operational semantics for the smart contract language Crystality. The semantics separates the system into engine components and a global component, making the modular structure of parallel smart contract execution explicit.

The compositional formulation enables clear statements and simple proofs of structural properties that are difficult to express in the original semantics. In particular, we prove locality of local steps, isolation of global execution, and strong commutativity of independent local steps.

We formally establish the correctness of this redesign by proving a transaction-level bisimulation between the compositional semantics and the original monolithic semantics. The proof is based on encoding and decoding functions relating configurations of the two semantics, together with code-level bisimulation results for local and global execution.

These results demonstrate that the compositional semantics provides a cleaner and more modular foundation for reasoning about parallel smart contracts. We believe that this framework can serve as a basis for future verification techniques and formal analyses for languages designed for parallel blockchain.

\begin{credits}
\subsubsection{\ackname} This research was sponsored by National Natural Science Foundation of China Grant No. 92582102, 62572013 and 62172019, and the National Key R\&D Program of China under Grant 2022YFB2702200.
\end{credits}

\bibliographystyle{splncs04}
\bibliography{mypaper}

\appendix

\section{Syntax of Crystality}

In the following, we give the whole syntax of Crystality, represented by a variant of Extended Backus-Naur Form (known as EBNF) where:

\begin{itemize}
 \item Terminal symbols are written in \tsym{monospaced fonts}.
 \item Non-terminal productions are encapsulated in  $\ntsym{angle brackets}$.
 \item Zero or one occurence is denoted by $^{?}$, zero or more occurences is denoted by $^{*}$.
\end{itemize}

A contract takes the following syntactic form:
 
\begin{bnf}
    \ntsym{contractDecl} ::= \tsym{contract} \ntsym{identifier} \tsym{\{} \ntsym{statevarDecl} ^{*} \ntsym{funcDecl}^{*} \tsym{\}}
\end{bnf}

\emph{Identifier} serves as a unique name for the contract. \emph{StatevarDecl} defines the various state variables associated with the contract, and \emph{funcDecl} specifies the functions within the contract.

\begin{bnf}
    \ntsym{statevarDecl} ::= \ntsym{type} \ntsym{scope} \ntsym{identifier} \tsym{;}
\end{bnf}

\begin{bnf}
    \ntsym{scope} ::= \tsym{@address} | \tsym{@engine} | \tsym{@global}
\end{bnf}

\emph{StatevarDecl} specifies both the type and the scope of the state variable.

The abstract syntax tree of functions is as follows.

\begin{bnf}
    \ntsym{funcDecl} ::=& \tsym{function}  \ntsym{identifier} \ntsym{funcPara} \ntsym{scope} \tsym{returns} \ntsym{type}^{?} \ntsym{funcBody} \\
    \ntsym{funcPara} ::=&  \tsym{(} (\ntsym{type}   \ntsym{identifier} \tsym{,} )^{*} \ntsym{type}   \ntsym{identifier} \tsym{)}  | \tsym{(} \tsym{)} \\
    \ntsym{funcBody} ::=& \tsym{\{} \ntsym{stmt} \tsym{\}} \\
    \ntsym{stmt} ::=& \ntsym{pstmt} | \ntsym{stmt} \ntsym{stmt} | \tsym{if} \tsym{(} \ntsym{exp} \tsym{)} \tsym{then} \tsym{\{} \ntsym{stmt} \tsym{\}} \tsym{else} \tsym{\{} \ntsym{stmt} \tsym{\}} \\ & | \tsym{while}  \tsym{(} \ntsym{exp} \tsym{)} \tsym{\{} \ntsym{stmt} \tsym{\}} \\
    \ntsym{pstmt} ::=& \ntsym{type} \ntsym{identifier} \tsym{;} | \tsym{skip} | \ntsym{identifier} := \ntsym{exp} \tsym{;} | \ntsym{relaycall} \\ & | \ntsym{returnstmt} | \ntsym{identifier} \tsym{(} \ntsym{exp} ^{*} \tsym{)} \\
    \ntsym{exp} ::=& \ntsym{identifier} | \ntsym{identifier} \tsym{(} \ntsym{exp} ^{*} \tsym{)} 
\end{bnf}

\begin{bnf}
    \ntsym{relaycall} ::= ( \tsym{relay @}\ntsym{exp} | \tsym{relay @ engines} | \tsym{relay @ global} ) \ntsym{identifier} \tsym{(} \ntsym{exp}^{*} \tsym{)} \tsym{;}
\end{bnf}

Basically, a function declaration consists of the following parts: its identifier, parameters along with their corresponding types, its scope, return type (if applicable), and the function body. The function body is enclosed within braces and consists of some statements. The statement may be primitive statement, sequential composition, conditional, or loop. A primitive statement may be temporary variable declaration, skip, assignment, relay call, return statement, or function call. The expression can be identifier or function call.

\section{Running Example}

\begin{lstlisting}[basicstyle=\ttfamily\small]
contract Bank {
  int @address balance;
  int @global total;
  function transfer(address payee, int amount) @address {
    if (amount <= balance) then {
      balance := balance - amount;
      relay @ payee deposit(amount);
      relay @ global updateTotal(amount);
    } else { skip }
  }
  function deposit(int amount) @address {
    balance := balance + amount;
  }
  function updateTotal(int amount) @global {
    total := total + amount;
  }
}
\end{lstlisting}

Assume there are two execution engines. Let $a,d$ be addresses managed by engine $1$, and $b,c$ be addresses managed by engine $2$.
Consider the following initial state:
\[
balance(a)=10,\quad balance(b)=5,\quad balance(c)=8,\quad balance(d)=1,
\]
and
\[
total = 0.
\]

Now suppose the current block contains exactly two user-initiated transactions:
\[
tx_1 = transfer(b,3) \text{ sent by } a,
\qquad
tx_2 = transfer(d,2) \text{ sent by } c.
\]

Let the initial system configuration be
\[
C_0 = (\Gamma^{(0)},\Omega^{(0)},\gamma_G^{(0)}),
\]
where
\[
\Gamma^{(0)} = (\gamma_1^{(0)},\gamma_2^{(0)}),\qquad
\Omega^{(0)} = (\emptyset,\emptyset,\emptyset),
\]
and $\gamma_G^{(0)}$ stores the initial global state with
\(
G(total)=0.
\)

Moreover, $\gamma_1^{(0)}$ contains the local program corresponding to $tx_1$,
and $\gamma_2^{(0)}$ contains the local program corresponding to $tx_2$.

\paragraph{First block: local execution.}

A transaction is executed in the engine that stores its sender's address. The transaction $tx_1$ is executed in component $\gamma_1$.
Its local effect is to decrease $balance(a)$ from $10$ to $7$.
At the same time, it emits two relay sets:
\[
\omega_1^{(1)}=\emptyset,\quad
\omega_2^{(1)}=\{(address,b,deposit(3))\},\quad
\omega_G^{(1)}=\{(global,updateTotal(3))\}.
\]
Thus, by one system-level local step, we obtain
\[
C_0=(\Gamma^{(0)},\Omega^{(0)},\gamma_G^{(0)})\Rightarrow_l
(\Gamma^{(1)},\Omega^{(1)},\gamma_G^{(0)})=C_1,
\]
where $\Gamma^{(1)}$ differs from $\Gamma^{(0)}$ only at the first component, and
\[
\Omega^{(1)}=(\emptyset,\{(address,b,deposit(3))\},\{(global,updateTotal(3))\}).
\]

Next, the transaction $tx_2$ is executed in component $\gamma_2$.
Its local effect is to decrease $balance(c)$ from $8$ to $6$.
It emits
\[
\omega_1^{(2)}=\{(address,d,deposit(2))\},\quad
\omega_2^{(2)}=\emptyset,\quad
\omega_G^{(2)}=\{(global,updateTotal(2))\}.
\]
Hence we obtain another system-level local step
\[
C_1=(\Gamma^{(1)},\Omega^{(1)},\gamma_G^{(0)})\Rightarrow_l
(\Gamma^{(2)},\Omega^{(2)},\gamma_G^{(0)})=C_2,
\]
where $\Gamma^{(2)}$ differs from $\Gamma^{(1)}$ only at the second component, and
\[
\Omega^{(2)}=
\left(
\begin{aligned}
&\{(address,d,deposit(2))\},\\
&\{(address,b,deposit(3))\},\\
&\{(global,updateTotal(3)),\,(global,updateTotal(2))\}
\end{aligned}
\right).
\]

At the end of the first block, the sender-side local updates have already been performed, while both the payee-side updates and the global aggregation are still pending in the mempools.

Before describing the second block, we clarify the scope of the semantics
considered in this paper.

Our semantics models only the execution of transactions \emph{within a block}. It does not formalize the process by which transactions are selected from mempools and assembled into the next block. In particular, the scheduling and ordering policies used by the blockchain protocol are outside the scope of the present semantics.

When the next block begins, we assume that the transactions pending in the mempools have already been collected and organized into the program fragments executed by the components of the system configuration.

Concretely, the global program is
\[
Prog_G = updateTotal(3);\,updateTotal(2),
\]
while the local programs for the two engines are
\[
Prog_1 = deposit(2), \qquad Prog_2 = deposit(3),
\]
where the corresponding address information is attached.

\paragraph{Second block: global phase followed by local phase.}

By the execution convention, all global transactions of a block are executed before any local transactions. Therefore, in the second block, global component executes first and update the global state. Hence we obtain
\[
C_2=(\Gamma^{(2)},\Omega^{(2)},\gamma_G^{(0)})\Rightarrow_g^{*}
(\Gamma^{(3)},\Omega^{(3)},\gamma_G^{(1)})=C_3,
\]

where
\(
G(total)=5
\)

After the global phase completes, engine $1$ and $2$ execute their pending transactions and update the payee balances and yields
\[
C_3=(\Gamma^{(3)},\Omega^{(3)},\gamma_G^{(1)})\Rightarrow_l^{*}
(\Gamma^{(4)},\Omega^{(4)},\gamma_G^{(1)})=C_4,
\]

Therefore, after the two blocks, the final state is
\begin{align*}
&balance(a)=7, \qquad\qquad balance(b)=8,\\
&balance(c)=6, \qquad\qquad balance(d)=3,\\
&total=5.
\end{align*}

This example illustrates four key aspects of the compositional semantics. 
First, each system-level local step changes exactly one local component
and leaves the global component unchanged.
Second, relay effects are not executed immediately, but are recorded explicitly in the corresponding mempools $\Omega_i$ and $\Omega_G$.
Third, the execution is naturally decomposed into a global phase and a local phase, matching the rule \textsc{GL}.
Finally, the two sender-side transfers are independent local computations,
so their order does not matter; this is precisely the kind of modular reasoning enabled by the compositional semantics and later formalized by our commutativity result. To summarize,
\[
C_0
\Rightarrow_l
C_1
\Rightarrow_l
C_2
\Rightarrow_g^{*}
C_3
\Rightarrow_l^{*}
C_4 .
\]

Here $C_0 \Rightarrow_l C_1 \Rightarrow_l C_2$ is the first block,
while $C_2 \Rightarrow_g^{*} C_3 \Rightarrow_l^{*} C_4$ is the second block.

\begin{corollary}[Commutativity in the running example]
In the first block, the executions of $tx_1$ and $tx_2$ commute:
executing $tx_1$ before $tx_2$, or $tx_2$ before $tx_1$, yields the same
configuration $C_2$.
\end{corollary}

\section{Complete Set of Semantic Rules}
\subsection{Semantic Rules of Statements}

\begin{gather*}
\drule{SVDa}{ N_{\Psi_{i,j}}(id) = \mathit{None}  ,\ 1\le j\le k \\ T_{\Psi_{i,j}'}(id) = Type ,\ 1\le j\le k  \\ N_{\Psi_{i,j}'}(id) = allocate\_new(Type,\Psi_{i,j}) ,\ 1\le j\le k\\ [N_{\Psi_{i,j}'}(id)]_{\Psi_{i,j}'}^{size(Type)} = init(Type) ,\ 1\le j\le k  }{ (\sigma_i, @address\ Type\ id;) ->_l^{\tau} (\sigma_i', \cdot) }
\end{gather*}

\begin{gather*}
\drule{SVDs}{ N_{\Psi_{i,s}}(id) = \mathit{None} \\ T_{\Psi_{i,s}'}(id) = Type \\\\ N_{\Psi_{i,s}'}(id) = allocate\_new(Type,\Psi_{i,s}) \\ [N_{\Psi_{i,s}'}(id)]_{\Psi_{i,s}'}^{size(Type)} = init(Type)  }{ (\sigma_i,@engine\ Type\ id; ) ->_l^{\tau} (\sigma_i',\cdot) }
\end{gather*}

\begin{gather*}
\drule{SVDg}{N_G(id) = \mathit{None}\\ T_{G'}(id) = Type \\\\ N_{G'}(id)=allocate\_new(Type,G)\\ [N_{G'}(id)_{G'}^{size(Type)}] = init(Type)  }{(\sigma_G,@global\ Type\ id;) ->_g^{\tau} (\sigma_G',\cdot)}
\end{gather*}

\begin{gather*}
\drule{TVDl}{top(M_i).scope \ne global \\ N_{top(M_i)}(id) = \mathit{None} \\ T_{top(M_i)'}(id) = Type \\ N_{top(M_i)'}(id) = allocate\_new(Type,top(M_i)) \\ [N_{top(M_i)'}(id)]_{top(M_i)'}^{size(Type)} = init(Type)}{(\sigma_i, Type\ id;) ->^{\tau}_l (\sigma_i',\cdot)}
\end{gather*}

\begin{gather*}
\drule{TVDg}{top(M_G).scope = global \\ N_{top(M_G)}(id) = \mathit{None} \\ T_{top(M_G)'}(id) = Type \\ N_{top(M_G)'}(id) = allocate\_new(Type,top(M_G)) \\ [N_{top(M_G)'}(id)]_{top(M_G)'}^{size(Type)} = init(Type)}{(\sigma_G, Type\ id;) ->^{\tau}_g (\sigma_G',\cdot)}
\end{gather*}

\begin{gather*}
\drule{TASSIGNl}{top(M_i).scope \ne global \\ \mathcal{E}[\![(\sigma_i, exp)]\!] ->_l^{?(\omega_1,...,\omega_n,\omega_G)} (\hat{\sigma_i},v) \\\\
N_{top(\hat{M_i})}(id) \ne \mathit{None} \\ [N_{top(\hat{M_i}')}(id)]^{size(T_{top(\hat{M_i}')}(id))}_{top(\hat{M_i}')}=v }{(\sigma_i,  id = exp;) ->_l^{?(\omega_1,...,\omega_n,\omega_G)} (\hat{\sigma_i}',\cdot)}
\end{gather*}

\begin{gather*}
\drule{TASSIGNg}{top(M_G).scope = global \\ \mathcal{E}[\![(\sigma_G, exp)]\!] ->_g^{?(\omega_1,...,\omega_n,\omega_G)} (\hat{\sigma_G},v) \\\\
N_{top(\hat{M_G})}(id) \ne \mathit{None} \\ [N_{top(\hat{M_G}')}(id)]^{size(T_{top(\hat{M_G}')}(id))}_{top(\hat{M_G}')}=v }{(\sigma_G,  id = exp;) ->_g^{?(\omega_1,...,\omega_n,\omega_G)} (\hat{\sigma_G}',\cdot)}
\end{gather*}

\begin{gather*}
\drule{SASSIGNaa}{\mathcal{E}[\![(\sigma_i,  exp)]\!] ->_l^{?(\omega_1,...,\omega_n,\omega_G)} (\hat{\sigma_i},v)   \\\\ top(M_i).scope = (address,j) \\ top(\hat{M_i}).scope = (address,j) \\\\ N_{\hat{\Psi_{i,j}}}(id) \ne \mathit{None} \\ [N_{\hat{\Psi_{i,j}}'}(id)]^{size(T_{\hat{\Psi_{i,j}}'}(id))}_{\hat{\Psi_{i,j}}'}=v }{(\sigma_i,  id = exp;) ->_l^{?(\omega_1,...,\omega_n,\omega_G)} (\hat{\sigma_i}',\cdot)}
\end{gather*}

\begin{gather*}
\drule{SASSIGNas}{\mathcal{E}[\![(\sigma_i,  exp)]\!] ->_l^{?(\omega_1,...,\omega_n,\omega_G)} (\hat{\sigma_i},v)    \\\\  top(M_i).scope = (address,j) \\ top(\hat{M_i}).scope = (address,j) \\\\ N_{\hat{\Psi_{i,s}}}(id) \ne \mathit{None} \\ [N_{\hat{\Psi_{i,s}}'}(id)]^{size(T_{\hat{\Psi_{i,s}}'}(id))}_{\hat{\Psi_{i,s}}'}=v }{(\sigma_i,  id = exp;) ->_l^{?(\omega_1,...,\omega_n,\omega_G)} (\hat{\sigma_i}',\cdot)}
\end{gather*}

\begin{gather*}
\drule{SASSIGNss}{\mathcal{E}[\![(\sigma_i,  exp)]\!] ->_l^{?(\omega_1,...,\omega_n,\omega_G)} (\hat{\sigma_i},v)    \\\\ top(M_i).scope = engine \\ top(\hat{M_i}).scope = engine \\\\  N_{\hat{\Psi_{i,s}}}(id) \ne \mathit{None} \\ [N_{\hat{\Psi_{i,s}}'}(id)]^{size(T_{\hat{\Psi_{i,s}}'}(id))}_{\hat{\Psi_{i,s}}'}=v }{(\sigma_i,  id = exp;) ->_l^{?(\omega_1,...,\omega_n,\omega_G)} (\hat{\sigma_i}',\cdot)}
\end{gather*}

\begin{gather*}
\drule{SASSIGNgg}{\mathcal{E}[\![(\sigma_G,  exp)]\!] ->_g^{?(\omega_1,...,\omega_n,\omega_G)} (\hat{\sigma_G},v)    \\\\ top(M_G).scope = global \\ top(\hat{M_G}).scope = global \\\\  N_{\hat{G}}(id) \ne \mathit{None} \\ [N_{\hat{G}'}(id)]^{size(T_{\hat{G}'}(id))}_{\hat{G}'}=v }{(\sigma_G,  id = exp;) ->_g^{?(\omega_1,...,\omega_n,\omega_G)} (\hat{\sigma_G}',\cdot)}
\end{gather*}

\begin{gather*}
\drule{IFUNaa}{\Lambda_{scope}(idf)=address \\ top(M_i).scope = (address,j)\\\\
\Lambda_{rettype}(idf) = \mathit{None} \\ \Lambda_{body}(idf) = Block \\\\
\Lambda_{paraname}(idf) =(id_1,...,id_t) \\ \Lambda_{paratype}(idf) = (T_1,...,T_t)\\\\ 
M_i^{(1)}=addtop(M_i,top(M_i)) \\ top(M_i^{(1)}).scope = (address,j) \\\\
top(M_i^{(1)}).rt = \mathit{None} \\ M_i^{(3)} = removetop(M_i^{(2)})\\\\
(\sigma_i^{(1)},T_1\ id_1 = exp_1;...;T_t\ id_t = exp_t;Block) ->_l^{?(\omega_1,...,\omega_n,\omega_G)} (\sigma_i^{(2)},\cdot) 
}{(\sigma_i,idf(exp_1,...,exp_t)) ->_l^{?(\omega_1,...,\omega_n,\omega_G)} (\sigma_i^{(3)},\cdot)}
\end{gather*}

\begin{gather*}
\drule{IFUNas}{\Lambda_{scope}(idf)=engine \\ top(M_i).scope = (address,j)\\\\
\Lambda_{rettype}(idf) = \mathit{None} \\ \Lambda_{body}(idf) = Block \\\\
\Lambda_{paraname}(idf) =(id_1,...,id_t) \\ \Lambda_{paratype}(idf) = (T_1,...,T_t)\\\\ 
M_i^{(1)}=addtop(M_i,top(M_i)) \\ top(M_i^{(1)}).scope = engine \\\\
top(M_i^{(1)}).rt = \mathit{None} \\ M_i^{(3)} = removetop(M_i^{(2)})\\\\
(\sigma_i^{(1)},T_1\ id_1 = exp_1;...;T_t\ id_t = exp_t;Block) ->_l^{?(\omega_1,...,\omega_n,\omega_G)} (\sigma_i^{(2)},\cdot) 
}{(\sigma_i,idf(exp_1,...,exp_t)) ->_l^{?(\omega_1,...,\omega_n,\omega_G)} (\sigma_i^{(3)},\cdot)}
\end{gather*}

\begin{gather*}
\drule{IFUNss}{\Lambda_{scope}(idf)=engine \\ top(M_i).scope = engine\\\\
\Lambda_{rettype}(idf) = \mathit{None} \\ \Lambda_{body}(idf) = Block \\\\
\Lambda_{paraname}(idf) =(id_1,...,id_t) \\ \Lambda_{paratype}(idf) = (T_1,...,T_t)\\\\
M_i^{(1)}=addtop(M_i,top(M_i)) \\ top(M_i^{(1)}).scope = engine \\\\
top(M_i^{(1)}).rt = \mathit{None} \\ M_i^{(3)} = removetop(M_i^{(2)})\\\\
(\sigma_i^{(1)},T_1\ id_1 = exp_1;...;T_t\ id_t = exp_t;Block) ->_l^{?(\omega_1,...,\omega_n,\omega_G)} (\sigma_i^{(2)},\cdot) 
}{(\sigma_i,idf(exp_1,...,exp_t)) ->_l^{?(\omega_1,...,\omega_n,\omega_G)} (\sigma_i^{(3)},\cdot)}
\end{gather*}

\begin{gather*}
\drule{IFUNgg}{\Lambda_{scope}(idf)=global \\ top(M_G).scope = global\\\\
\Lambda_{rettype}(idf) = \mathit{None} \\ \Lambda_{body}(idf) = Block \\\\
\Lambda_{paraname}(idf) =(id_1,...,id_t) \\ \Lambda_{paratype}(idf) = (T_1,...,T_t)\\\\
M_G^{(1)}=addtop(M_G,top(M_G))\\
top(M_G^{(1)}).scope = global \\\\
top(M_G^{(1)}).rt = \mathit{None} \\
M_G^{(3)} = removetop(M_G^{(2)}) \\\\
(\sigma_G^{(1)},T_1\ id_1 = exp_1;...;T_t\ id_t = exp_t;Block) ->_g^{?(\omega_1,...,\omega_n,\omega_G)} (\sigma_G^{(2)},\cdot) }{(\sigma_G,idf(exp_1,...,exp_t)) ->_g^{?(\omega_1,...,\omega_n,\omega_G)} (\sigma_G^{(3)},\cdot)}
\end{gather*}

\begin{gather*}
\drule{RELAYal}{top(M_i).scope \ne global \\ \Lambda_{scope}(idf) = address \\\\ \mathcal{E}[\![(\sigma_i,exp)]\!] ->_l^{\tau} (\sigma_i,(r,j)) \\  \mathcal{E}[\![(\sigma_i,exp_l)]\!] ->_l^{\tau} (\sigma_i,v_l), 1 \le l \le t  \\\\
\omega_r= \{ (address,j,idf(v_1,...,v_t)) \} \\
\omega_\iota = \emptyset , \forall \iota \ne r }{(\sigma_i,relay@exp\ idf(exp_1,...,exp_t);) ->_l^{?(\omega_1,...,\omega_n,\omega_G)} (\sigma_i,\cdot)}
\end{gather*}

\begin{gather*}
\drule{RELAYag}{top(M_G).scope = global \\ \Lambda_{scope}(idf) = address \\\\ \mathcal{E}[\![(\sigma_G,exp)]\!] ->_g^{\tau} (\sigma_G,(r,j)) \\       \mathcal{E}[\![(\sigma_G,exp_l)]\!] ->_g^{\tau} (\sigma_G,v_l), 1 \le l \le t   \\\\
\omega_r= \{ (address,j,idf(v_1,...,v_t)) \} \\
\omega_\iota = \emptyset , \forall \iota \ne r
}{(\sigma_G,relay@exp\ idf(exp_1,...,exp_t);) ->_g^{?(\omega_1,...,\omega_n,\omega_G)} (\sigma_G,\cdot)}
\end{gather*}

\begin{gather*}
\drule{RELAYsl}{top(M_i).scope \ne global\\ \Lambda_{scope}(idf) = engine \\\\  \mathcal{E}[\![(\sigma_i,exp_l)]\!] ->_l^{\tau} (\sigma_i,v_l), 1 \le l \le t \\
\omega_i = \{ (engine,idf(v_1,...,v_t)) \} ,\ 1 \le i \le n\\ \omega_G = \emptyset}{(\sigma_i,relay@engines\ idf(exp_1,...,exp_t);) ->_l^{?(\omega_1,...,\omega_n,\omega_G)} (\sigma_i,\cdot)}
\end{gather*}

\begin{gather*}
\drule{RELAYsg}{top(M_G).scope = global\\ \Lambda_{scope}(idf) = engine \\\\ \mathcal{E}[\![(\sigma_G,exp_l)]\!] ->_g^{\tau} (\sigma_G,v_l), 1 \le l \le t \\
\omega_i = \{ (engine,idf(v_1,...,v_t)) \} ,\ 1 \le i \le n\\ \omega_G = \emptyset}{(\sigma_G,relay@engines\ idf(exp_1,...,exp_t);) ->_g^{?(\omega_1,...,\omega_n,\omega_G)} (\sigma_G,\cdot)}
\end{gather*}

\begin{gather*}
\drule{RELAYgl1}{\Lambda_{scope}(idf) = global \\ top(M_i).scope = (address,j) \\\\  \mathcal{E}[\![(\sigma_i,exp_l)]\!] ->_l^{\tau} (\sigma_i,v_l), 1 \le l \le t \\
\omega_G =  \{ (global,idf(v_1,...,v_t)) \}  \\ \omega_\iota = \emptyset , \forall \iota \ne G}{(\sigma_i,relay@global\ idf(exp_1,...,exp_t);) ->_l^{?(\omega_1,...,\omega_n,\omega_G)} (\sigma_i,\cdot)}
\end{gather*}

\begin{gather*}
\drule{RELAYgl2}{\Lambda_{scope}(idf) = global \\   top(M_i).scope = engine  \\\\  \mathcal{E}[\![(\sigma_i,exp_l)]\!] ->_l^{\tau} (\sigma_i,v_l), 1 \le l \le t \\
\omega_G =  \{ (global,idf(v_1,...,v_t)) \} \\ \omega_\iota = \emptyset , \forall \iota \ne G}{(\sigma_i,relay@global\ idf(exp_1,...,exp_t);) ->_l^{?(\omega_1,...,\omega_n,\omega_G)} (\sigma_i,\cdot)}
\end{gather*}

\begin{gather*}
\drule{RELAYo}{\quad}{(\Omega_1,...,\Omega_n,\Omega_G) ->^{!(\omega_1,...,\omega_n,\omega_G)} (\Omega_1 \cup \omega_1,...,\Omega_n \cup \omega_n,\Omega_G \cup \omega_G)}
\end{gather*}

\begin{gather*}
\drule{SEQl}{(\sigma_i,Prog_1;)->_l^{?(\hat{\omega_1},...,\hat{\omega_n},\hat{\omega_G})}  (\sigma_i',\cdot)\\ (\sigma_i',Prog_2;)->_l^{?(\bar{\omega_1},...,\bar{\omega_n},\bar{\omega_G})} (\sigma_i'',\cdot) \\ \omega_\iota = \hat{\omega_\iota} \cup \bar{\omega_\iota}   , \forall \iota  }{(\sigma_i,Prog_1;Prog_2;) ->_l^{?(\omega_1,...,\omega_n,\omega_G)} (\sigma_i'',\cdot)}
\end{gather*}

\begin{gather*}
\drule{SEQg}{(\sigma_G,Prog_1;)->_g^{?(\hat{\omega_1},...,\hat{\omega_n},\hat{\omega_G})}  (\sigma_G',\cdot)\\ (\sigma_G',Prog_2;)->_g^{?(\bar{\omega_1},...,\bar{\omega_n},\bar{\omega_G})} (\sigma_G'',\cdot) \\ \omega_\iota = \hat{\omega_\iota} \cup \bar{\omega_\iota}   , \forall \iota  }{(\sigma_G,Prog_1;Prog_2;) ->_g^{?(\omega_1,...,\omega_n,\omega_G)} (\sigma_G'',\cdot)}
\end{gather*}

\subsection{Semantic Rules of Evaluations} 

\begin{gather*}
\drule{ETIDl}{ top(M_i).scope \ne global \\ N_{top(M_i)}(id) \ne \mathit{None} \\
[N_{top(M_i)}(id)]^{size(T_{top(M_i)}(id))}_{top(M_i)}=v }{\mathcal{E}[\![(\sigma_i,id) ]\!]->_l^{\tau}(\sigma_i,v)}
\end{gather*}

\begin{gather*}
\drule{ETIDg}{ top(M_G).scope = global \\ N_{top(M_G)}(id) \ne \mathit{None} \\
[N_{top(M_G)}(id)]^{size(T_{top(M_G)}(id))}_{top(M_G)}=v }{\mathcal{E}[\![(\sigma_G,id) ]\!]->_g^{\tau}(\sigma_G,v)}
\end{gather*}

\begin{gather*}
\drule{ESIDaa}{top(M_i).scope = (address,j) \\ N_{\Psi_{i,j}}(id) \ne \mathit{None}\\
[N_{\Psi_{i,j}}(id)]^{size(T_{\Psi_{i,j}}(id))}_{\Psi_{i,j}}=v }{\mathcal{E}[\![(\sigma_i,id) ]\!]->_l^{\tau}(\sigma_i,v)}
\end{gather*}

\begin{gather*}
\drule{ESIDas}{top(M_i).scope = (address,j) \\ N_{\Psi_{i,s}}(id) \ne \mathit{None}\\
[N_{\Psi_{i,s}}(id)]^{size(T_{\Psi_{i,s}}(id))}_{\Psi_{i,s}}=v }{\mathcal{E}[\![(\sigma_i,id) ]\!]->_l^{\tau}(\sigma_i,v)}
\end{gather*}

\begin{gather*}
\drule{ESIDag}{top(M_i).scope = (address,j) \\ N_{G_i}(id) \ne \mathit{None} \\ [N_{G_i}(id)]^{size(T_{G_i}(id))}_{G_i}=v }{\mathcal{E}[\![(\sigma_i,id) ]\!]->_l^{\tau}(\sigma_i,v)}
\end{gather*}

\begin{gather*}
\drule{ESIDss}{top(M_i).scope = engine \\ N_{\Psi_{i,s}}(id) \ne \mathit{None}\\
[N_{\Psi_{i,s}}(id)]^{size(T_{\Psi_{i,s}}(id))}_{\Psi_{i,s}}=v }{\mathcal{E}[\![(\sigma_i,id) ]\!]->_l^{\tau}(\sigma_i,v)}
\end{gather*}

\begin{gather*}
\drule{ESIDsg}{top(M_i).scope = engine \\ N_{G_i}(id) \ne \mathit{None} \\ [N_{G_i}(id)]^{size(T_{G_i}(id))}_{G_i}=v }{\mathcal{E}[\![(\sigma_i,id) ]\!]->_l^{\tau}(\sigma_i,v)}
\end{gather*}

\begin{gather*}
\drule{ESIDgg}{top(M_G).scope = global \\ N_G(id) \ne \mathit{None}\\
[N_G(id)]^{size(T_G(id))}_G=v }{\mathcal{E}[\![(\sigma_G,id) ]\!]->_g^{\tau}(\sigma_G,v)}
\end{gather*}

\begin{gather*}
\drule{EFUNaa}{\Lambda_{scope}(idf)=address \\ top(M_i).scope = (address,j)\\\\
\Lambda_{rettype}(idf) = Type \\ \Lambda_{body}(idf) = Block \\\\
\Lambda_{paraname}(idf)= (id_1,...,id_t) \\ \Lambda_{paratype}(idf) = (T_1,...,T_t)\\\\
M_i^{(1)}=addtop(M_i,top(M_i)) \\ top(M_i^{(1)}).scope = (address,j)\\\\
ret = new\_ID(top(M_i^{(1)})) \\ top(M_i^{(1)}).rt = ret\\\\
(\sigma_i^{(1)},T_1\ id_1 = exp_1;...;T_t\ id_t = exp_t;Type\ ret; Block) ->_l^{?(\omega_1,...,\omega_n,\omega_G)} (\sigma_i^{(2)},\cdot) \\\\
M_i^{(3)} = removetop(M_i^{(2)}) \\[N_{top(M_i^{(2)})}(ret)]^{size(Type)}_{top(M_i^{(2)})} = v} {\mathcal{E}[\![(\sigma_i,idf(exp_1,...,exp_t))]\!] ->_l^{?(\omega_1,...,\omega_n,\omega_G)} (\sigma_i^{(3)},v)}
\end{gather*}

\begin{gather*}
\drule{EFUNas}{\Lambda_{scope}(idf)=engine \\ top(M_i).scope = (address,j)\\\\
\Lambda_{rettype}(idf) = Type \\ \Lambda_{body}(idf) = Block \\\\
\Lambda_{paraname}(idf)= (id_1,...,id_t) \\ \Lambda_{paratype}(idf) = (T_1,...,T_t)\\\\
M_i^{(1)}=addtop(M_i,top(M_i)) \\ top(M_i^{(1)}).scope = engine\\\\
ret = new\_ID(top(M_i^{(1)})) \\ top(M_i^{(1)}).rt = ret\\\\
(\sigma_i^{(1)},T_1\ id_1 = exp_1;...;T_t\ id_t = exp_t;Type\ ret; Block) ->_l^{?(\omega_1,...,\omega_n,\omega_G)} (\sigma_i^{(2)},\cdot) \\\\
M_i^{(3)} = removetop(M_i^{(2)}) \\[N_{top(M_i^{(2)})}(ret)]^{size(Type)}_{top(M_i^{(2)})} = v} {\mathcal{E}[\![(\sigma_i,idf(exp_1,...,exp_t))]\!] ->_l^{?(\omega_1,...,\omega_n,\omega_G)} (\sigma_i^{(3)},v)}
\end{gather*}

\begin{gather*}
\drule{EFUNss}{\Lambda_{scope}(idf)=engine \\ top(M_i).scope = engine\\\\
\Lambda_{rettype}(idf) = Type \\ \Lambda_{body}(idf) = Block \\\\
\Lambda_{paraname}(idf)= (id_1,...,id_t) \\ \Lambda_{paratype}(idf) = (T_1,...,T_t)\\\\
M_i^{(1)}=addtop(M_i,top(M_i)) \\ top(M_i^{(1)}).scope = engine\\\\
ret = new\_ID(top(M_i^{(1)})) \\ top(M_i^{(1)}).rt = ret\\\\
(\sigma_i^{(1)},T_1\ id_1 = exp_1;...;T_t\ id_t = exp_t;Type\ ret; Block) ->_l^{?(\omega_1,...,\omega_n,\omega_G)} (\sigma_i^{(2)},\cdot) \\\\
M_i^{(3)} = removetop(M_i^{(2)}) \\[N_{top(M_i^{(2)})}(ret)]^{size(Type)}_{top(M_i^{(2)})} = v} {\mathcal{E}[\![(\sigma_i,idf(exp_1,...,exp_t))]\!] ->_l^{?(\omega_1,...,\omega_n,\omega_G)} (\sigma_i^{(3)},v)}
\end{gather*}

\begin{gather*}
\drule{EFUNgg}{\Lambda_{scope}(idf)=global \\ top(M_G).scope = global \\\\
\Lambda_{rettype}(idf) = Type \\ \Lambda_{body}(idf) = Block \\\\
\Lambda_{paraname}(idf)= (id_1,...,id_t) \\ \Lambda_{paratype}(idf) = (T_1,...,T_t)\\\\
M_G^{(1)}=addtop(M_G,top(M_G))  \\ top(M_G^{(1)}).scope = global  \\\\
ret = new\_ID(top(M_G^{(1)}))  \\ top(M_G^{(1)}).rt = ret  \\\\
(\sigma_G^{(1)},T_1\ id_1 = exp_1;...;T_t\ id_t = exp_t;Type\ ret; Block) ->_g^{?(\omega_1,...,\omega_n,\omega_G)}  (\sigma_G^{(2)},\cdot) \\\\
[N_{top(M_G^{(2)})}(ret)]^{size(Type)}_{top(M_G^{(2)})} = v  \\
M_G^{(3)} = removetop(M_G^{(2)}) } {\mathcal{E}[\![(\sigma_G,idf(exp_1,...,exp_t))]\!] ->_g^{?(\omega_1,...,\omega_n,\omega_G)}  (\sigma_G^{(3)},v)}
\end{gather*}

\subsection{Semantic Rules of Transactions}

\begin{gather*}
\drule{IFUNta}{\Lambda_{scope}(idf)=address \\ top(M_i).scope = \mathit{None}\\\\
\Lambda_{rettype}(idf) = \mathit{None} \\ \Lambda_{body}(idf) = Block \\\\
\Lambda_{paraname}(idf) =(id_1,...,id_t) \\ \Lambda_{paratype}(idf) = (T_1,...,T_t)\\\\
(i,j) = get\_sender\_address() \\ top(M_i^{(1)}).scope = (address,j) \\\\ M_i^{(1)}=addtop(M_i,\{ \}) \\ 
top(M_i^{(1)}).rt = \mathit{None} \\ M_i^{(3)} = removetop(M_i^{(2)})\\\\
(\sigma_i^{(1)},T_1\ id_1 = exp_1;...;T_t\ id_t = exp_t;Block) ->_l^{?(\omega_1,...,\omega_n,\omega_G)} (\sigma_i^{(2)},\cdot) 
}{(\sigma_i,idf(exp_1,...,exp_t)) ->_{l,t}^{?(\omega_1,...,\omega_n,\omega_G)} (\sigma_i^{(3)},\cdot)}
\end{gather*}

\begin{gather*}
\drule{IFUNts}{\Lambda_{scope}(idf)=engine \\ top(M_i).scope = \mathit{None}\\\\
\Lambda_{rettype}(idf) = \mathit{None} \\ \Lambda_{body}(idf) = Block \\\\
\Lambda_{paraname}(idf) =(id_1,...,id_t) \\ \Lambda_{paratype}(idf) = (T_1,...,T_t)\\\\
(i,j) = get\_sender\_address() \\ top(M_i^{(1)}).scope = engine \\\\
M_i^{(1)}=addtop(M_i,\{ \}) \\ top(M_i^{(1)}).rt = \mathit{None} \\ M_i^{(3)} = removetop(M_i^{(2)})\\\\
(\sigma_i^{(1)},T_1\ id_1 = exp_1;...;T_t\ id_t = exp_t;Block) ->_l^{?(\omega_1,...,\omega_n,\omega_G)} (\sigma_i^{(2)},\cdot) 
}{(\sigma_i,idf(exp_1,...,exp_t)) ->_{l,t}^{?(\omega_1,...,\omega_n,\omega_G)} (\sigma_i^{(3)},\cdot)}
\end{gather*}

\begin{gather*}
\drule{IFUNtg}{\Lambda_{scope}(idf)=global \\ top(M_G).scope = \mathit{None}\\\\
\Lambda_{rettype}(idf) = \mathit{None} \\ \Lambda_{body}(idf) = Block \\\\
\Lambda_{paraname}(idf) =(id_1,...,id_t) \\ \Lambda_{paratype}(idf) = (T_1,...,T_t)\\\\
M_G^{(1)}=addtop(M_G,\{ \})\\
top(M_G^{(1)}).scope = global \\\\
top(M_G^{(1)}).rt = \mathit{None} \\
M_G^{(3)} = removetop(M_G^{(2)}) \\\\
(\sigma_G^{(1)},T_1\ id_1 = exp_1;...;T_t\ id_t = exp_t;Block) ->_g^{?(\omega_1,...,\omega_n,\omega_G)} (\sigma_G^{(2)},\cdot) }{(\sigma_G,idf(exp_1,...,exp_t)) ->_{g,t}^{?(\omega_1,...,\omega_n,\omega_G)} (\sigma_G^{(3)},\cdot)}
\end{gather*}

\begin{gather*}
\drule{EFUNta}{\Lambda_{scope}(idf)=address \\ top(M_i).scope = \mathit{None}\\\\
\Lambda_{rettype}(idf) = Type \\ \Lambda_{body}(idf) = Block \\\\
\Lambda_{paraname}(idf)= (id_1,...,id_t) \\ \Lambda_{paratype}(idf) = (T_1,...,T_t)\\ \\
(i,j) = get\_sender\_address() \\ top(M_i^{(1)}).scope = (address,j)\\\\ 
ret = new\_ID(top(M_i^{(1)})) \\ top(M_i^{(1)}).rt = ret\\\\
(\sigma_i^{(1)},T_1\ id_1 = exp_1;...;T_t\ id_t = exp_t;Type\ ret; Block) ->_l^{?(\omega_1,...,\omega_n,\omega_G)} (\sigma_i^{(2)},\cdot) \\\\
M_i^{(1)}=addtop(M_i,\{ \})  \\ M_i^{(3)} = removetop(M_i^{(2)}) \\[N_{top(M_i^{(2)})}(ret)]^{size(Type)}_{top(M_i^{(2)})} = v} {\mathcal{E}[\![(\sigma_i,idf(exp_1,...,exp_t))]\!] ->_{l,t}^{?(\omega_1,...,\omega_n,\omega_G)} (\sigma_i^{(3)},v)}
\end{gather*}

\begin{gather*}
\drule{EFUNts}{\Lambda_{scope}(idf)=engine \\ top(M_i).scope = \mathit{None}\\\\
\Lambda_{rettype}(idf) = Type \\ \Lambda_{body}(idf) = Block \\\\
\Lambda_{paraname}(idf)= (id_1,...,id_t) \\ \Lambda_{paratype}(idf) = (T_1,...,T_t)\\\\
(i,j) = get\_sender\_address() \\ top(M_i^{(1)}).scope = engine\\\\
ret = new\_ID(top(M_i^{(1)})) \\ top(M_i^{(1)}).rt = ret\\\\
(\sigma_i^{(1)},T_1\ id_1 = exp_1;...;T_t\ id_t = exp_t;Type\ ret; Block) ->_l^{?(\omega_1,...,\omega_n,\omega_G)} (\sigma_i^{(2)},\cdot) \\\\
M_i^{(1)}=addtop(M_i,\{ \})  \\ M_i^{(3)} = removetop(M_i^{(2)}) \\[N_{top(M_i^{(2)})}(ret)]^{size(Type)}_{top(M_i^{(2)})} = v} {\mathcal{E}[\![(\sigma_i,idf(exp_1,...,exp_t))]\!] ->_{l,t}^{?(\omega_1,...,\omega_n,\omega_G)} (\sigma_i^{(3)},v)}
\end{gather*}

\begin{gather*}
\drule{EFUNtg}{\Lambda_{scope}(idf)=global \\ top(M_G).scope = \mathit{None} \\\\
\Lambda_{rettype}(idf) = Type \\ \Lambda_{body}(idf) = Block \\\\
\Lambda_{paraname}(idf)= (id_1,...,id_t) \\ \Lambda_{paratype}(idf) = (T_1,...,T_t)\\\\
M_G^{(1)}=addtop(M_G,\{ \})  \\ top(M_G^{(1)}).scope = global  \\\\
ret = new\_ID(top(M_G^{(1)}))  \\ top(M_G^{(1)}).rt = ret  \\\\
(\sigma_G^{(1)},T_1\ id_1 = exp_1;...;T_t\ id_t = exp_t;Type\ ret; Block) ->_g^{?(\omega_1,...,\omega_n,\omega_G)}  (\sigma_G^{(2)},\cdot) \\\\
[N_{top(M_G^{(2)})}(ret)]^{size(Type)}_{top(M_G^{(2)})} = v  \\
M_G^{(3)} = removetop(M_G^{(2)}) } {\mathcal{E}[\![(\sigma_G,idf(exp_1,...,exp_t))]\!] ->_{g,t}^{?(\omega_1,...,\omega_n,\omega_G)}  (\sigma_G^{(3)},v)}
\end{gather*}

\begin{gather*}
\drule{SEQtl}{(\sigma_i,Prog_1) ->_{l,t}^{?(\omega_1,...,\omega_n,\omega_G)} (\sigma_i',\cdot)}{(\sigma_i,Prog_1; Prog_2) ->_{l,t}^{?(\omega_1,...,\omega_n,\omega_G)} (\sigma_i',Prog_2)}
\end{gather*}

\begin{gather*}
\drule{SEQtg}{(\sigma_G,Prog_1) ->_{g,t}^{?(\omega_1,...,\omega_n,\omega_G)} (\sigma_G',\cdot)}{(\sigma_G,Prog_1; Prog_2) ->_{g,t}^{?(\omega_1,...,\omega_n,\omega_G)} (\sigma_G',Prog_2)}
\end{gather*}

\subsection{Semantic Rules of the System} 

\begin{gather*}
\drule{TAUg1}{\gamma_G ->_g^{\tau} \gamma_G' \\ \omega_\iota = \emptyset , \forall \iota }{\gamma_G ->_g^{?(\omega_1,...,\omega_n,\omega_G)} \gamma_G'}
\\[0.5em]
\drule{TAUg2}{\delta_G ->_g^{\tau} \gamma_G' \\ \omega_\iota = \emptyset , \forall \iota }{\delta_G ->_g^{?(\omega_1,...,\omega_n,\omega_G)} \gamma_G'}
\\[0.5em]
\drule{TAUl1}{\gamma_i ->_l^{\tau} \gamma_i' \\ \omega_\iota = \emptyset , \forall \iota }{\gamma_i ->_l^{?(\omega_1,...,\omega_n,\omega_G)} \gamma_i'}
\\[0.5em]
\drule{TAUl2}{\delta_i ->_l^{\tau} \gamma_i' \\ \omega_\iota = \emptyset , \forall \iota }{\delta_i ->_l^{?(\omega_1,...,\omega_n,\omega_G)} \gamma_i'}
\end{gather*}

\begin{gather*}
\drule{G}{\gamma_G ->_{g,t}^{?(\omega_1,...,\omega_n,\omega_G)} \gamma_G' \\  \Omega ->^{!(\omega_1,...,\omega_n,\omega_G)} \Omega' \\\\ \rho(\Gamma,\Gamma') \\ \mu(\Gamma,\Omega,\gamma_G) \\ \mu(\Gamma',\Omega',\gamma_G')}{(\Gamma,\Omega,\gamma_G) \sysstep_g (\Gamma',\Omega',\gamma_G') }
\end{gather*}

\begin{gather*}
\drule{G1}{\gamma_G ->_{g,t}^{\tau} \gamma_G'\\ \rho(\Gamma,\Gamma') \\ \mu(\Gamma,\Omega,\gamma_G) \\ \mu(\Gamma',\Omega,\gamma_G')}{(\Gamma,\Omega,\gamma_G) \sysstep_g (\Gamma',\Omega,\gamma_G') }
\end{gather*}

\begin{gather*}
\drule{L}{\gamma_i ->_{l,t}^{?(\omega_1,...,\omega_n,\omega_G)} \gamma_i' \\  \Omega ->^{!(\omega_1,...,\omega_n,\omega_G)} \Omega'\\
\Gamma=(\gamma_1,...,\gamma_i,...,\gamma_n)\\ \Gamma'=(\gamma_1,...,\gamma_i',...,\gamma_n) \\ \gamma_G = (\sigma_G,\cdot) \\ \mu(\Gamma,\Omega,\gamma_G) \\ \mu(\Gamma',\Omega',\gamma_G)}{(\Gamma,\Omega,\gamma_G) \sysstep_l (\Gamma',\Omega',\gamma_G)}
\end{gather*}

\begin{gather*}
\drule{L1}{\gamma_i ->_{l,t}^{\tau} \gamma_i' \\
\Gamma=(\gamma_1,...,\gamma_i,...,\gamma_n)\\ \Gamma'=(\gamma_1,...,\gamma_i',...,\gamma_n) \\ \gamma_G = (\sigma_G,\cdot) \\ \mu(\Gamma,\Omega,\gamma_G) \\ \mu(\Gamma',\Omega,\gamma_G)}{(\Gamma,\Omega,\gamma_G) \sysstep_l (\Gamma',\Omega,\gamma_G)}
\end{gather*}

\begin{gather*}
\drule{Ga}{ \mu(\Gamma,\Omega,\gamma_G)  }{(\Gamma,\Omega,\gamma_G) \sysstep_g^{*} (\Gamma,\Omega,\gamma_G) }
\\[0.5em]
\drule{Gs}{(\Gamma,\Omega,\gamma_G) \sysstep_g (\Gamma',\Omega',\gamma_G')  \\ (\Gamma',\Omega',\gamma_G') \sysstep_g^{*} (\Gamma'',\Omega'',\gamma_G'') }{(\Gamma,\Omega,\gamma_G) \sysstep_g^{*} (\Gamma'',\Omega'',\gamma_G'')}
\\[0.5em]
\drule{La}{ \mu(\Gamma,\Omega,\gamma_G) }{(\Gamma,\Omega,\gamma_G) \sysstep_l^{*} (\Gamma,\Omega,\gamma_G)}
\\[0.5em]
\drule{Ls}{(\Gamma,\Omega,\gamma_G) \sysstep_l (\Gamma',\Omega',\gamma_G) \\ (\Gamma',\Omega',\gamma_G) \sysstep_l^{*} (\Gamma'',\Omega'',\gamma_G) }{(\Gamma,\Omega,\gamma_G) \sysstep_l^{*} (\Gamma'',\Omega'',\gamma_G)}
\end{gather*}

\begin{gather*}
\drule{Para}{\gamma_i ->_{l,t}^{?(\hat{\omega_1},...,\hat{\omega_n},\hat{\omega_G})} \gamma_i'\\ \gamma_j ->_{l,t}^{?(\bar{\omega_1},...,\bar{\omega_n},\bar{\omega_G})} \gamma_j' \\ \gamma_G = (\sigma_G,\cdot) \\ \omega_\iota = \hat{\omega_\iota} \cup \bar{\omega_\iota}   , \forall \iota \\  \Omega ->^{!(\omega_1,...,\omega_n,\omega_G)} \Omega'\\
\Gamma=(\gamma_1,...,\gamma_i,...,\gamma_j,...,\gamma_n)\\ \Gamma'=(\gamma_1,...,\gamma_i',...,\gamma_j',...,\gamma_n) \\ \mu(\Gamma,\Omega,\gamma_G) \\ \mu(\Gamma',\Omega',\gamma_G)}{(\Gamma,\Omega,\gamma_G) \sysstep_l^{*} (\Gamma',\Omega',\gamma_G)}
\end{gather*}

\begin{gather*}
\drule{GL}{(\Gamma,\Omega,\gamma_G) \sysstep_g^{*} (\Gamma',\Omega',\gamma_G')  \\ (\Gamma',\Omega',\gamma_G') \sysstep_l^{*} (\Gamma'',\Omega'',\gamma_G') }{(\Gamma,\Omega,\gamma_G) \sysstep_{gl}^{*} (\Gamma'',\Omega'',\gamma_G')}
\end{gather*}

\section{Properties and Proofs}

\begin{proof}[Proof of Theorem~\ref{thm:locality}]
By inversion on the system rule used to derive the step.
The only applicable rule is \(L\), which updates exactly one component
\(\gamma_i\), leaves \(\gamma_G\) unchanged, and extends the mempools by
union with the emitted relay transactions.
\end{proof}

\begin{proof}[Proof of Theorem~\ref{thm:global-isolation}]
By inversion on the system rule deriving the step.
The only applicable rule is \(G\), which updates the global component
and synchronizes the copies of the global state in \(\Gamma\), while the mempools grow by union with the
emitted relay transactions.
\end{proof}

\begin{proof}[Proof of Theorem~\ref{thm:strong-comm-local}]
By Theorem~\ref{thm:locality}, each local system step
updates exactly one engine component, leaves the global component
unchanged, and extends the mempools by union with the emitted relay
transactions. Since \(i\neq j\), the two steps affect different engine
components. Hence applying the \(i\)-step and then the \(j\)-step yields
the same engine components, the same global component, and the same
mempool as applying the \(j\)-step and then the \(i\)-step. Therefore
the two steps commute, and the resulting configuration \(C'\) is unique.
\end{proof}

\begin{theorem}[Mempool monotonicity]
\label{thm:mempool-monotonicity}
For any system execution
\[
(\Gamma,\Omega,\gamma_G)\sysstep^{*}(\Gamma',\Omega',\gamma_G'),
\]
we have
\[
\Omega \subseteq \Omega'.
\]
\end{theorem}

\begin{proof}
By induction on the derivation of \(\sysstep^{*}\).
In each one-step system transition, the mempool changes only through the
relay-receiving operation, which extends it by set union (See Rule RELAYo). Hence the mempool is monotonically increasing in one step. The multi-step result then follows immediately from the closure rules.
\end{proof}

\begin{theorem}[Associativity of mempool reception]
\label{thm:relay-reception-union}
If
\[
\Omega \to^{!(\hat{\omega}_1,\ldots,\hat{\omega}_n,\hat{\omega}_G)} \Omega'
\qquad\text{and}\qquad
\Omega' \to^{!(\bar{\omega}_1,\ldots,\bar{\omega}_n,\bar{\omega}_G)} \Omega'',
\]
and
\[
\omega_{\iota}=\hat{\omega}_{\iota}\cup \bar{\omega}_{\iota}
\qquad\text{for every } \iota\in\{1,\ldots,n,G\},
\]
then
\[
\Omega \to^{!(\omega_1,\ldots,\omega_n,\omega_G)} \Omega''.
\]
\end{theorem}

\begin{proof}
Immediate from Rule RELAYo. By definition of the receiving step,
the first transition updates each mempool component by union with
\(\hat{\omega}_{\iota}\), and the second by union with
\(\bar{\omega}_{\iota}\). This is exactly the same as a single receiving step with \(\omega_{\iota}=\hat{\omega}_{\iota}\cup \bar{\omega}_{\iota}\) for each \(\iota\).
\end{proof}

\section{Proofs of the Bisimulation Results}

\paragraph{Notation convention.}
Throughout the bisimulation section, symbols from the original
(non-compositional) semantics are annotated with the superscript
\(\mathsf{o}\), while unannotated symbols refer to the compositional semantics. In particular, we write \(C^{\mathsf{o}}\), \(\sigma_i^{\mathsf{o}}\), \(\Omega_i^{\mathsf{o}}\),
\(Prog_i^{\mathsf{o}}\), \(G^{\mathsf{o}}\) and \(\to^{\mathsf{o}}\) for the corresponding notions in the original semantics.

\begin{definition}[Global prefix and local suffix]
Let \(Prog\) be a program consisting of a sequence of
T-functions of the form
\[
idf_1(\ldots) ; idf_2(\ldots) ; \cdots ; idf_n(\ldots),
\]
where \(n \ge 0\).

In Crystality, T-functions whose scope is \textsf{global}
must appear before those whose scope is \textsf{address} or
\textsf{engine}. Hence every such program \(Prog\) can be uniquely
decomposed as
\[
Prog = Prog_g ; Prog_l,
\]
where every T-function in \(Prog_g\) has scope \textsf{global},
and every T-function in \(Prog_l\) has scope \textsf{address} or
\textsf{engine}.

We define
\[
\mathsf{gprefix}(Prog) = Prog_g,
\qquad
\mathsf{lsuffix}(Prog) = Prog_l.
\]
\end{definition}

\begin{lemma}
For any program \(Prog\) consisting of T-functions,
\[
Prog = \mathsf{gprefix}(Prog) \; ; \; \mathsf{lsuffix}(Prog).
\]
\end{lemma}

\begin{lemma}[Empty-stack invariant at transaction boundaries]
In both the original semantics and the compositional semantics,
all temporary stacks are empty at transaction boundaries.
That is, immediately before and after the execution of a
transaction, \(M_i = \emptyset\) for every engine \(i\),
and \(M_G = \emptyset\).
\end{lemma}

\begin{proof}
The property follows by inspection of the semantic rules.
Function-call rules push a frame onto the corresponding stack,
and the matching return rules pop the frame before the transaction
terminates. Hence all stacks are empty at transaction boundaries.
\end{proof}

\begin{definition}[Mempool projection operations]
For each engine mempool \(\Omega_i\), we write
\[
\mathsf{GTx}(\Omega_i)
\;=\;
\{\, (global,\, idf(...)) \in \Omega_i \mid idf(...) \text{ is a T-function call} \,\}
\]
for the set of \textsf{global} transactions contained in \(\Omega_i\). And we write
\[
\mathsf{LTx}(\Omega_i)
\;=\;
\Omega_i \setminus \mathsf{GTx}(\Omega_i).
\]
for the set of \textsf{address} or \textsf{engine} transactions contained in \(\Omega_i\). 
\end{definition}

\begin{lemma}[Agreement of global transactions in mempools] \label{mempool_lemma}
For any reachable configuration in the original semantics, the global
transactions contained in the mempools of all engines are identical.
Formally, for all \(1 \le i,j \le n\),
\[
\mathsf{GTx}(\Omega_i)=\mathsf{GTx}(\Omega_j).
\]
\end{lemma}

\begin{proof}
By inspection of the relay rules. Global transactions cannot be
user-initiated; they are generated only by \texttt{relay@global}.
Moreover, every such relay is broadcast to all engines, so the same
global transaction tuple is added to every mempool simultaneously.
Hence the sets of global transactions in all \(\Omega_i\) are identical.
\end{proof}

\begin{remark}[Decomposition of mempools]
By the previous lemma, the set of global transactions extracted from
each engine mempool is the same. Each engine mempool \(\Omega_i\) can be decomposed as
\[
\Omega_i = \mathsf{LTx}(\Omega_i) \cup \mathsf{GTx}(\Omega_i),
\]
\end{remark}

\begin{lemma}[Well-formedness of reachable configurations in the compositional semantics]
Every reachable configuration in the compositional semantics is
well formed. That is, if
\[
(\Gamma,\Omega,\gamma_G) \;\sysstep^{*}\; (\Gamma',\Omega',\gamma_G'),
\]
then
\[
\mathsf{wf}(\Gamma',\Omega',\gamma_G').
\]
\end{lemma}

\begin{proof}
By induction on the derivation of the reachability relation.
It suffices to observe that every system-level step, whether $\sysstep_l^{*}$ or $\sysstep_g^{*}$, preserves the predicate $\mu(\Gamma,\Omega,\gamma_G) $ and thus the equality
\(
G_1 = \cdots = G_n = G
\).
Hence all reachable configurations are well formed.
\end{proof}

\begin{lemma}[Well-formedness of reachable configurations in the original semantics]
Every reachable configuration in the original semantics is
well formed. That is, if
\[
C^{\mathsf{o}} \;{\to^{\mathsf{o}}}^{*}\; C^{\mathsf{o}\prime},
\]
then
\[
\mathsf{wf}^{\mathsf{o}}(C^{\mathsf{o}\prime}).
\]
\end{lemma}

\begin{proof}
By inspection of the semantic rules of the original semantics.
Global transactions are executed jointly by all engines. Consequently, \(\mathsf{gprefix}(Prog_1^{\mathsf{o}})=\cdots=
\mathsf{gprefix}(Prog_n^{\mathsf{o}})\) is preserved throughout execution.
\end{proof}

\begin{lemma}[Agreement of local stacks during global execution] \label{mi_lemma}
Let
\[
C^{\mathsf{o}}
=
(\sigma_1^{\mathsf{o}},\Omega_1^{\mathsf{o}},Prog_1^{\mathsf{o}},\ldots,
\sigma_n^{\mathsf{o}},\Omega_n^{\mathsf{o}},Prog_n^{\mathsf{o}},G^{\mathsf{o}})
\]
be a well-formed configuration of the original semantics, where
\[
\sigma_i^{\mathsf{o}}=(\Psi_i^{\mathsf{o}},M_i^{\mathsf{o}})
\qquad (1 \le i \le n).
\]
If
\[
\mathsf{gprefix}(Prog_i^{\mathsf{o}}) \neq \cdot
\]
for some \(i\) (and hence for all \(i\), by well-formedness), then
\[
M_1^{\mathsf{o}} = M_2^{\mathsf{o}} = \cdots = M_n^{\mathsf{o}}.
\]
\end{lemma}

\begin{proof}
Since \(C^{\mathsf{o}}\) is well formed, we have
\[
\mathsf{gprefix}(Prog_1^{\mathsf{o}})
=
\mathsf{gprefix}(Prog_2^{\mathsf{o}})
=
\cdots
=
\mathsf{gprefix}(Prog_n^{\mathsf{o}}).
\]
If this common global prefix is non-empty, then the system is currently
executing global transactions. In the original semantics, global
functions are executed jointly by all engines using the same code, and
their execution is synchronized by the consensus mechanism. Therefore,
the temporary stacks maintained by all engines evolve identically during
such execution, which yields
\[
M_1^{\mathsf{o}} = M_2^{\mathsf{o}} = \cdots = M_n^{\mathsf{o}}.
\]
\end{proof}

\begin{remark}
The definition of \(Enc\) is well defined.
Since the original configuration is well formed,
\(\mathsf{gprefix}(Prog_i^{\mathsf{o}})\) is identical for all engines.
Moreover, by Lemma~\ref{mempool_lemma} the sets
\(\mathsf{GTx}(\Omega_i^{\mathsf{o}})\) coincide,
and by Lemma~\ref{mi_lemma} the stacks \(M_i^{\mathsf{o}}\) coincide whenever
\(\mathsf{gprefix}(Prog_i^{\mathsf{o}})\neq\cdot\).
Hence the construction above does not depend on the choice of \(i\).
\end{remark}

\begin{remark}
The definition of \(Dec\) is well defined. Since the configuration of the compositional semantics is well formed, all engines share the same global state \(G\).
\end{remark}

\begin{proof}[Proof of Lemma~\ref{lem:wf-preserve-enc-dec}]
Both statements follow directly from the definitions.

For \(Enc\), in the encoded configuration all local states use the same global state \(G^{\mathsf{o}}\), hence the resulting configuration satisfies \(\mathsf{wf}\).

For \(Dec\), the decoded programs are \(Prog_i^{\mathsf{o}} = Prog_G ; Prog_i\). Since \(Prog_G\) consists only of \textsf{global} T-functions and
\(Prog_i\) consists only of \textsf{address} or \textsf{engine}
T-functions, we have \(\mathsf{gprefix}(Prog_G ; Prog_i) = Prog_G\). Thus all decoded programs share the same global prefix, and the resulting configuration satisfies \(\mathsf{wf}^{\mathsf{o}}\).
\end{proof}

\begin{lemma}[Empty local stacks during the global phase]
\label{lem:empty-local-stacks-global-phase}
Let \((\Gamma,\Omega,\gamma_G)\) be a reachable configuration of the
compositional semantics, where
\[
\Gamma=(\gamma_1,\ldots,\gamma_n), \qquad
\gamma_i=((\Psi_i,M_i,G),Prog_i), \qquad
\gamma_G=((G,M_G),Prog_G).
\]
If \(Prog_G \neq \cdot\), then
\[
M_1 = M_2 = \cdots = M_n = \emptyset .
\]
\end{lemma}

\begin{proof}
A reachable configuration with \(Prog_G \neq \cdot\) is in the global
phase of transactions execution. By the system-level execution discipline, the global transactions in \(Prog_G\) are executed before any local
computation starts. Hence no local step has been performed yet, so none
of the local stacks \(M_i\) has been used. Since the local stacks are
initially empty, they remain empty throughout this phase.
\end{proof}

\begin{lemma}[Empty global stack during the local phase]
\label{lem:empty-global-stack-local-phase}
Let \((\Gamma,\Omega,\gamma_G)\) be a reachable configuration of the
compositional semantics, where
\[
\Gamma=(\gamma_1,\ldots,\gamma_n), \qquad
\gamma_i=((\Psi_i,M_i,G),Prog_i), \qquad
\gamma_G=((G,M_G),Prog_G).
\]
If \(Prog_G = \cdot\), then
\[
M_G = \emptyset .
\]
\end{lemma}

\begin{proof}
If \(Prog_G=\cdot\), then all global transactions of the current block
have already been completed, or there were none to begin with. The
global stack is empty at transaction boundaries, and once the global
phase has finished, the subsequent local phase does not modify the
global component. Therefore \(M_G\) is empty.
\end{proof}

\begin{proof}[Proof of Lemma~\ref{thm:enc-dec-inverse}]
We prove the two statements separately.

\smallskip
\noindent
\emph{(1)} Let
\[
C^{\mathsf{o}} =
(\sigma_1^{\mathsf{o}},\Omega_1^{\mathsf{o}},Prog_1^{\mathsf{o}},\ldots,
\sigma_n^{\mathsf{o}},\Omega_n^{\mathsf{o}},Prog_n^{\mathsf{o}},G^{\mathsf{o}})
\]
be a well-formed configuration of the original semantics, where
\(\sigma_i^{\mathsf{o}}=(\Psi_i^{\mathsf{o}},M_i^{\mathsf{o}})\).
We distinguish two cases.

If \(\mathsf{gprefix}(Prog_i^{\mathsf{o}})\neq\cdot\), then by the
definition of \(Enc\), all local stacks in the encoded configuration are
empty and the global stack is \(M_i^{\mathsf{o}}\). Decoding this
configuration restores the programs, mempools, and states immediately.
For the stack component, it suffices to use Lemma~\ref{mi_lemma} stating that, in a well-formed configuration of the original semantics with non-empty global prefix, all stacks \(M_i^{\mathsf{o}}\) are identical. Hence
\(Dec(Enc(C^{\mathsf{o}})) = C^{\mathsf{o}}\).

If \(\mathsf{gprefix}(Prog_i^{\mathsf{o}})=\cdot\), then by the
definitions of \(Enc\) and \(Dec\), the two transformations leave all
components unchanged, and again \(Dec(Enc(C^{\mathsf{o}})) = C^{\mathsf{o}}\).

\smallskip
\noindent
\emph{(2)} Let
\[
C=(\Gamma,\Omega,\gamma_G)
\]
be a reachable well-formed configuration of the compositional semantics,
where
\[
\Gamma=(\gamma_1,\ldots,\gamma_n), \qquad
\gamma_i=((\Psi_i,M_i,G),Prog_i), \qquad
\gamma_G=((G,M_G),Prog_G).
\]
Again we distinguish two cases.

If \(Prog_G\neq\cdot\), then by the definition of \(Dec\) and \(Enc\), decoding followed by encoding restores all components immediately except
potentially the local stacks. In this case, by Lemma~\ref{lem:empty-local-stacks-global-phase} on empty local stacks during the global phase, we have
\(M_1 = \cdots = M_n = \emptyset\). Therefore the reconstructed local stacks are exactly the original ones, and hence \( Enc(Dec(C)) = C \).

If \(Prog_G=\cdot\), decoding followed by encoding restores all components immediately except potentially the global stack. In this case, by Lemma~\ref{lem:empty-global-stack-local-phase} on the empty global stack during the local phase, we have \( M_G=\emptyset \). Therefore the reconstructed global stack is exactly the original one, and thus \(Enc(Dec(C)) = C \).

This completes the proof.
\end{proof}

\begin{definition}[Transaction-boundary configurations]
A configuration is said to be at a \emph{transaction boundary}
if all memory stacks are empty.

For a configuration of the compositional semantics,
this means
\[
M_i = \emptyset \ (1 \le i \le n)
\quad\text{and}\quad
M_G = \emptyset .
\]

For a configuration of the original semantics,
this means
\[
M_i^{\mathsf{o}} = \emptyset \ (1 \le i \le n).
\]
\end{definition}

\begin{definition}[Transaction-level encoding and decoding]
Let \(Enc_T\) and \(Dec_T\) denote the restrictions of
\(Enc\) and \(Dec\) to configurations at transaction boundaries.

Since all stacks are empty at transaction boundaries,
the stack components in the encoded and decoded configurations
are defined to be empty, while all other components are
constructed exactly as in Definitions~\ref{def:enc} and
\ref{def:dec}. No case distinction is needed.
\end{definition}

\begin{lemma}
For any configuration at a transaction boundary,
\[
Enc_T = Enc,
\qquad
Dec_T = Dec .
\]
\end{lemma}

\begin{proof}
Immediate from the definitions, since all memory stacks are empty
at transaction boundaries.
\end{proof}

\begin{corollary}
All lemmas and theorems established for \(Enc\) and \(Dec\)
remain valid for \(Enc_T\) and \(Dec_T\).
\end{corollary}

\begin{definition}[Code-level correspondence]
Let \(Conf^{\mathsf{o}}\) and \(Conf\) denote the sets of configurations
of the original semantics and the compositional semantics,
respectively.

We define a relation
\[
R \subseteq Conf^{\mathsf{o}} \times Conf
\]
such that \(R(C^{\mathsf{o}},C)\) holds iff

\begin{enumerate}
\item \(C^{\mathsf{o}}\) is reachable and well formed in the original semantics;
\item \(C\) is reachable and well formed in the compositional semantics; and
\item \(Enc(C^{\mathsf{o}})=C\)
      (equivalently, \(Dec(C)=C^{\mathsf{o}}\)).
\end{enumerate}
\end{definition}

This relation captures the correspondence between the two semantics
at arbitrary code boundaries.

\begin{lemma}
If \(R_T(C^{\mathsf{o}},C)\), then \(R(C^{\mathsf{o}},C)\).
\end{lemma}

\begin{proof}
Immediate from the definitions, since \(R_T\) strengthens \(R\) by
additionally requiring \(C^{\mathsf{o}}\) and \(C\) to be at
transaction boundaries and using the restricted mappings
\(Enc_T\) and \(Dec_T\), which coincide with \(Enc\) and \(Dec\) on
such configurations.
\end{proof}

\begin{definition}[Local wrapping]
Let \(i \in \{1,\ldots,n\}\).
We write
\[
Wrap_l(i,\gamma_i,C)
\]
if \(C=(\Gamma,\Omega,\gamma_G)\) is a reachable and well-formed
configuration of the compositional semantics such that

\begin{enumerate}
\item the \(i\)-th component of \(\Gamma\) is \(\gamma_i\); and
\item the global program is empty, i.e., \(Prog_G=\cdot\).
\end{enumerate}
\end{definition}

\begin{definition}[Local correspondence]
Let \(i \in \{1,\ldots,n\}\).
We write
\[
R_l(i,\gamma_i,C^{\mathsf{o}})
\]
if there exists a configuration \(C\) such that

\begin{enumerate}
\item \(Wrap_l(i,\gamma_i,C)\);
\item \(R(C^{\mathsf{o}},C)\); and
\item for every engine \(j\), \( \mathsf{gprefix}(Prog_j^{\mathsf{o}})=\cdot\).
\end{enumerate}
\end{definition}

\begin{remark}
Intuitively, the relation \(R_l\) isolates the behavior of a
single engine \(i\) during the local phase.
The condition \(Wrap_l(i,\gamma_i,C)\) embeds the component
\(\gamma_i\) into a well-formed system configuration in which
no global transactions remain to be executed.
The relation \(R(C^{\mathsf{o}},C)\) then establishes the
code-level correspondence between the two semantics.

Note that the global states coincide.
Since \(C\) is well formed, all local copies \(G_i\) are equal to
the global state \(G\).
Moreover, \(R(C^{\mathsf{o}},C)\) implies that the global state of
\(C^{\mathsf{o}}\) is also \(G\).
Hence the global state in \(C^{\mathsf{o}}\) coincides with the
local copy \(G_i\) of the component \(\gamma_i\).
\end{remark}

\begin{definition}[Mempool difference in the original semantics]
Let
\[
\Delta\Omega = (\Delta\Omega_1,\ldots,\Delta\Omega_n).
\]
For two configurations \(C^{\mathsf{o}}\) and \(C^{\mathsf{o}\prime}\) of the
original semantics, we write
\[
\mathsf{MemDiff}^{\mathsf{o}}(C^{\mathsf{o}},C^{\mathsf{o}\prime},\Delta\Omega)
\]
if, for every \(1 \le i \le n\),
\[
\Omega_i^{\mathsf{o}\prime}
=
\Omega_i^{\mathsf{o}} \cup \Delta\Omega_i.
\]
\end{definition}

\begin{definition}[Local agreement outside the active component]
Let
\[
C^{\mathsf{o}}
=
(\sigma_1^{\mathsf{o}},\Omega_1^{\mathsf{o}},Prog_1^{\mathsf{o}},\ldots,
\sigma_n^{\mathsf{o}},\Omega_n^{\mathsf{o}},Prog_n^{\mathsf{o}},G^{\mathsf{o}})
\]
and
\[
C^{\mathsf{o}\prime}
=
(\sigma_1^{\mathsf{o}\prime},\Omega_1^{\mathsf{o}\prime},Prog_1^{\mathsf{o}\prime},\ldots,
\sigma_n^{\mathsf{o}\prime},\Omega_n^{\mathsf{o}\prime},Prog_n^{\mathsf{o}\prime},G^{\mathsf{o}\prime})
\]
be two configurations of the original semantics. For \(i\in\{1,\ldots,n\}\) we write
\[
\mathsf{Same}_l^{\mathsf{o}}(i,C^{\mathsf{o}},C^{\mathsf{o}\prime})
\]
if for every \(j\neq i\)
\[
\sigma_j^{\mathsf{o}\prime}=\sigma_j^{\mathsf{o}}
\qquad\text{and}\qquad
Prog_j^{\mathsf{o}\prime}=Prog_j^{\mathsf{o}}.
\]
\end{definition}

\begin{remark}
The relation \(\mathsf{Same}_l^{\mathsf{o}}\) expresses that all engines except \(i\) remain unchanged, while the mempools, the \(i\)-th engine,
and the global state may differ.
\end{remark}

Since the statement in the main text is compressed due to space limitations, we restate Theorem~\ref{lem:local-code-bisim} here in a clearer and more expanded format.

\medskip
\noindent\textbf{Theorem~\ref{lem:local-code-bisim} (Restated).}
If \(R_l(i,\gamma_i,C^{\mathsf{o}})\), then the following properties hold.

\begin{enumerate}
\item[(i)] If
\[
\gamma_i \to_l^{?(w_1,\ldots,w_n,w_G)} \gamma_i',
\]
then there exists a configuration \(C^{\mathsf{o}\prime}\) such that
\[
C^{\mathsf{o}} \to C^{\mathsf{o}\prime},
\]
\[
\mathsf{Same}_l^{\mathsf{o}}(i,C^{\mathsf{o}},C^{\mathsf{o}\prime}),
\]
\[
\mathsf{MemDiff}^{\mathsf{o}}(C^{\mathsf{o}},C^{\mathsf{o}\prime},\Delta\Omega),
\]
where
\[
\Delta\Omega=(w_1\cup w_G,\; w_2\cup w_G,\; \ldots,\; w_n\cup w_G),
\]
and
\[
R_l(i,\gamma_i',C^{\mathsf{o}\prime}).
\]

\item[(ii)] If
\[
C^{\mathsf{o}} \to C^{\mathsf{o}\prime},
\qquad
\mathsf{Same}_l^{\mathsf{o}}(i,C^{\mathsf{o}},C^{\mathsf{o}\prime}),
\qquad
\mathsf{MemDiff}^{\mathsf{o}}(C^{\mathsf{o}},C^{\mathsf{o}\prime},\Delta\Omega),
\]
where
\[
\Delta\Omega=(\Delta\Omega_1,\ldots,\Delta\Omega_n),
\]
then there exists a component \(\gamma_i'\) such that
\[
\gamma_i \to_l^{?(w_1,\ldots,w_n,w_G)} \gamma_i',
\]
where
\[
w_j = \mathsf{LTx}(\Delta\Omega_j)\quad (1\le j\le n),
\]
and
\[
w_G = \mathsf{GTx}(\Delta\Omega_j)
\]
for any \(j \in \{1,\ldots,n\}\), and
\[
R_l(i,\gamma_i',C^{\mathsf{o}\prime}).
\]
\end{enumerate}

\begin{remark}
In item (ii), the definition of \(w_G\) is well defined.
By Lemma~\ref{mempool_lemma}, the sets \(\mathsf{GTx}(\Delta\Omega_j)\) are identical for all \(j\). Hence \(w_G\) does not depend on the choice of \(j\).
\end{remark}

\begin{remark}
Theorem~\ref{lem:local-code-bisim} captures the code-level correspondence
between one local component in the compositional semantics and the
matching engine in the original semantics during the local phase.
The relation \(R_l\) isolates engine \(i\) while keeping the remaining
engines abstracted inside a well-formed system configuration.
Part~(i) shows that every local step of \(\gamma_i\) can be simulated by
a step of the original semantics, with the same effect on the active
engine and with the generated relay information reflected in the mempool
difference. Conversely, part~(ii) shows that every such step of the
original semantics can be represented by a corresponding local step of
\(\gamma_i\). This lemma is the key code-level ingredient for proving
the transaction-level bisimulation theorem.
\end{remark}

We now establish the code-level bisimulation for global execution.
While the previous lemma concerns the behavior of a single engine during
the local phase, the following result relates the execution of the
global component in the compositional semantics with the corresponding
global transactions in the original semantics.

\begin{definition}[Global wrapping]
We write
\[
Wrap_g(\bar{\gamma_G},C)
\]
if \(C=(\Gamma,\Omega,\gamma_G)\) is a reachable and well-formed
configuration of the compositional semantics such that

\begin{enumerate}
\item the global component of \(C\) is \(\bar{\gamma_G}\), i.e., \(\gamma_G=\bar{\gamma_G}\) ;
\item the global program is non-empty, i.e., \(Prog_G \neq \cdot\); and
\item for every local component \(\gamma_i\) in \(\Gamma\), its program is
empty, i.e., \(Prog_i = \cdot\).
\end{enumerate}
\end{definition}

\begin{definition}[Global correspondence]
We write
\[
R_g(\gamma_G,C^{\mathsf{o}})
\]
if there exists a configuration \(C\) such that

\begin{enumerate}
\item \(Wrap_g(\gamma_G,C)\);
\item \(R(C^{\mathsf{o}},C)\); and
\item for every engine \(i\), \(\mathsf{gprefix}(Prog_i^{\mathsf{o}})\neq\cdot \).
\end{enumerate}
\end{definition}

\begin{remark}
The relation \(Wrap_g\) isolates the global phase of execution.
It embeds the global component \(\gamma_G\) into a well-formed system
configuration in which the global program is still pending, while all
local components are inactive.
\end{remark}

\begin{remark}
Intuitively, the relation \(R_g\) matches a global component in the
compositional semantics with a configuration of the original semantics
during the global phase.
Since \(R(C^{\mathsf{o}},C)\) holds, the global state of
\(C^{\mathsf{o}}\) coincides with the global state carried by
\(\gamma_G\).
Moreover, the pending global program \(Prog_G\) and the stack \(M_G\)
represented in \(\gamma_G\) correspond to the \(Prog_i^{\mathsf{o}}\) and
\(M_i^{\mathsf{o}}\) distributed across all engines in the original
semantics.
\end{remark}

\begin{definition}[Global agreement on local states]
Let
\[
C^{\mathsf{o}}
=
(\sigma_1^{\mathsf{o}},\Omega_1^{\mathsf{o}},Prog_1^{\mathsf{o}},\ldots,
\sigma_n^{\mathsf{o}},\Omega_n^{\mathsf{o}},Prog_n^{\mathsf{o}},G^{\mathsf{o}})
\]
and
\[
C^{\mathsf{o}\prime}
=
(\sigma_1^{\mathsf{o}\prime},\Omega_1^{\mathsf{o}\prime},Prog_1^{\mathsf{o}\prime},\ldots,
\sigma_n^{\mathsf{o}\prime},\Omega_n^{\mathsf{o}\prime},Prog_n^{\mathsf{o}\prime},G^{\mathsf{o}\prime})
\]
be two configurations of the original semantics.
We write
\[
\mathsf{Same}_g^{\mathsf{o}}(C^{\mathsf{o}},C^{\mathsf{o}\prime})
\]
if, for every \(i \in \{1,\ldots,n\}\),
\[
\Psi_i^{\mathsf{o}\prime} = \Psi_i^{\mathsf{o}}.
\]
\end{definition}

\begin{remark}
Recall that \(\sigma_i^{\mathsf{o}}=(\Psi_i^{\mathsf{o}},M_i^{\mathsf{o}})\).
The relation \(\mathsf{Same}_g^{\mathsf{o}}\) expresses that the local states of all engines remain unchanged. Other components, including memory stacks, mempools, programs, and the global state, may differ.
\end{remark}

Since the statement in the main text is compressed due to space limitations, we restate Theorem~\ref{lem:global-code-bisim} here in a clearer and more expanded format.

\medskip
\noindent\textbf{Theorem~\ref{lem:global-code-bisim} (Restated).}
If \(R_g(\gamma_G,C^{\mathsf{o}})\), then the following properties hold.

\begin{enumerate}
\item[(i)] If
\[
\gamma_G \to_g^{?(w_1,\ldots,w_n,w_G)} \gamma_G',
\]
then there exists a configuration \(C^{\mathsf{o}\prime}\) such that
\[
C^{\mathsf{o}} \to C^{\mathsf{o}\prime},
\]
\[
\mathsf{Same}_g^{\mathsf{o}}(C^{\mathsf{o}},C^{\mathsf{o}\prime}),
\]
\[
\mathsf{MemDiff}^{\mathsf{o}}(C^{\mathsf{o}},C^{\mathsf{o}\prime},\Delta\Omega),
\]
where
\[
\Delta\Omega=(w_1\cup w_G,\; w_2\cup w_G,\; \ldots,\; w_n\cup w_G),
\]
and
\[
R_g(\gamma_G',C^{\mathsf{o}\prime}).
\]

\item[(ii)] If
\[
C^{\mathsf{o}} \to C^{\mathsf{o}\prime},
\qquad
\mathsf{Same}_g^{\mathsf{o}}(C^{\mathsf{o}},C^{\mathsf{o}\prime}),
\qquad
\mathsf{MemDiff}^{\mathsf{o}}(C^{\mathsf{o}},C^{\mathsf{o}\prime},\Delta\Omega),
\]
where
\[
\Delta\Omega=(\Delta\Omega_1,\ldots,\Delta\Omega_n),
\]
then there exists a component \(\gamma_G'\) such that
\[
\gamma_G \to_g^{?(w_1,\ldots,w_n,w_G)} \gamma_G',
\]
where
\[
w_i = \mathsf{LTx}(\Delta\Omega_i) \quad (1\le i\le n),
\]
and
\[
w_G = \mathsf{GTx}(\Delta\Omega_i)
\]
for any \(i \in \{1,\ldots,n\}\), and
\[
R_g(\gamma_G',C^{\mathsf{o}\prime}).
\]
\end{enumerate}

\begin{remark}
In item (ii), the definition of \(w_G\) is well defined.
By Lemma~\ref{mempool_lemma}, the sets \(\mathsf{GTx}(\Delta\Omega_i)\) are identical for all \(i\). Hence \(w_G\) does not depend on the choice of \(i\).
\end{remark}

\begin{remark}
Theorem~\ref{lem:global-code-bisim} establishes the correspondence between
the execution of the global component in the compositional semantics and
the execution of global transactions in the original semantics.
The relation \(R_g\) isolates the global phase of execution.
Part~(i) shows that every global step of \(\gamma_G\) can be simulated by
a step of the original semantics with the same effect on the mempools,
while preserving the local states of all engines.
Conversely, part~(ii) shows that every such step of the original
semantics can be represented by a corresponding global step of
\(\gamma_G\).
\end{remark}

\begin{proof}[Proof of Theorem~\ref{lem:local-code-bisim} and~\ref{lem:global-code-bisim}]
The two code-level bisimulation lemmas are proved together.
The underlying principle is that the code-level rules of the compositional
semantics correspond one by one to the code-level rules of the original
semantics.
The proof proceeds by induction on the derivation of the given step in one
semantics, combined with inversion on the corresponding rule in the other
semantics.
In each case, we show that the two rules perform the same essential state
update, produce the same residual program, and generate compatible mempool
effects. The preservation of the correspondence relations \(R_l\) and \(R_g\)
then follows by the definitions of \(Enc\), \(Dec\), \(\mathsf{Same}_l\),
\(\mathsf{Same}_g\), and \(\mathsf{MemDiff}^{\mathsf{o}}\).

To make the proof structure explicit, we group all code-level rules into the following five classes.

\smallskip
\noindent
\textbf{Class A: structural rules for common program constructs.}
This class contains the rules for sequential composition, conditional, and loop,
namely \textsc{SEQ}, \textsc{COND}, and \textsc{WHILE}.
These cases are routine: the proof follows immediately from the induction
hypothesis on the immediate sub-derivation(s), since the two semantics only
repackage the same control-flow step in different configuration formats.

\smallskip
\noindent
\textbf{Class B: variable declaration and assignment rules.}
This class contains the declaration and assignment rules, including
\textsc{SVD}, \textsc{TVD}, \textsc{SASSIGN}, and \textsc{TASSIGN}.
Here the key point is that both semantics update exactly the same storage or
stack component, while the encoding/decoding functions only reorganize the
surrounding configuration.
We present \textsc{SVDa}, \textsc{SVDg}, \textsc{SASSIGNaa}, and
\textsc{SASSIGNgg} as representative cases.

\smallskip
\noindent
\textbf{Class C: function-call rules.}
This class contains \textsc{IFUN} and \textsc{EFUN}.
These cases are handled by expanding the call, aligning the newly created stack
frame(s), and applying the induction hypothesis to the execution of the function
body.
The crucial invariant is that local calls remain in the local phase, whereas
global calls remain in the global phase.
We present \textsc{IFUNaa} and \textsc{IFUNgg} as representative cases.

\smallskip
\noindent
\textbf{Class D: relay rules.}
This class contains the \textsc{RELAY} rules.
These are the most important cases for the bisimulation argument, because they
explain how the relay labels
\(?(w_1,\ldots,w_n,w_G)\)
in the compositional semantics correspond to mempool growth in the original
semantics.
The local and global correspondence lemmas use
\(\mathsf{MemDiff}^{\mathsf{o}}\)
precisely to record this effect.
We present \textsc{RELAYal} and \textsc{RELAYgl} as representative cases.

\smallskip
\noindent
\textbf{Class E: evaluation rules.}
This class contains all evaluation rules, including \textsc{ETID} and
\textsc{ESID}.
These cases are straightforward, since they only read from the appropriate
state component and do not change the control structure.
Their proofs are either immediate or entirely analogous to the representative
rules from the previous classes.

\smallskip
\noindent
In the main text, we only present the representative cases
\textsc{SVDa}, \textsc{SVDg}, \textsc{SASSIGNaa}, \textsc{SASSIGNgg},
\textsc{IFUNaa}, \textsc{IFUNgg}, \textsc{RELAYal}, and \textsc{RELAYgl}.
The remaining cases are either trivial or proved in exactly the same manner.
This organization makes explicit where compositionality matters:
Class~A is purely structural, Classes~B and~C show that local and global code
manipulate the same computational content under different configuration
decompositions, and Class~D isolates the treatment of asynchronous relay
effects, which is the main semantic difference between ordinary local execution
and cross-engine interaction.

\paragraph{Case \textsc{SVDa} (local).}
We show the local code-level bisimulation for rule \textsc{SVDa} in the
compositional semantics and rule \textsc{SDa} in the original semantics.

Suppose \(R_l(i,\gamma_i,C^{\mathsf{o}})\) and
\[
\gamma_i \to_l^{\tau} \gamma_i'
\]
is derived by rule \textsc{SVDa}. By definition of \(R_l\), there exists a
reachable well-formed system configuration \(C\) such that
\(Wrap_l(i,\gamma_i,C)\) and \(R(C^{\mathsf{o}},C)\).
Since \textsc{SVDa} and \textsc{SDa} have exactly the same premises above the
line, they perform the same update on the corresponding \(@address\) storage
component of engine \(i\), and leave all other relevant components unchanged.
Moreover, no relay transaction is generated, so the mempool difference is empty.
Hence there exists \(C^{\mathsf{o}\prime}\) such that
\[
C^{\mathsf{o}} \to C^{\mathsf{o}\prime},
\qquad
\mathsf{Same}_l(i,C^{\mathsf{o}},C^{\mathsf{o}\prime}),
\qquad
\mathsf{MemDiff}^{\mathsf{o}}(C^{\mathsf{o}},C^{\mathsf{o}\prime},(\emptyset,\ldots,\emptyset)),
\]
and \(R_l(i,\gamma_i',C^{\mathsf{o}\prime})\) holds immediately.

Conversely, suppose
\[
C^{\mathsf{o}} \to C^{\mathsf{o}\prime}
\]
is derived by rule \textsc{SDa}, together with
\[
\mathsf{Same}_l(i,C^{\mathsf{o}},C^{\mathsf{o}\prime}),
\qquad
\mathsf{MemDiff}^{\mathsf{o}}(C^{\mathsf{o}},C^{\mathsf{o}\prime},\Delta\Omega).
\]
Since \textsc{SDa} generates no relay transaction, necessarily
\(\Delta\Omega=(\emptyset,\ldots,\emptyset)\).
Again, the premises of \textsc{SDa} and \textsc{SVDa} coincide exactly, so
there exists \(\gamma_i'\) such that
\[
\gamma_i \to_l^{\tau} \gamma_i'
\]
by \textsc{SVDa}, and \(R_l(i,\gamma_i',C^{\mathsf{o}\prime})\) follows
immediately from the definitions.

\paragraph{Case \textsc{SVDg} (global).}
We show the global code-level bisimulation for rule \textsc{SVDg} in the
compositional semantics and rule \textsc{SDg} in the original semantics.

Suppose \(R_g(\gamma_G,C^{\mathsf{o}})\) and
\[
\gamma_G \to_g^{\tau} \gamma_G'
\]
is derived by rule \textsc{SVDg}. By definition of \(R_g\), there exists a
reachable well-formed system configuration \(C\) such that
\(Wrap_g(\gamma_G,C)\) and \(R(C^{\mathsf{o}},C)\).
The rules \textsc{SVDg} and \textsc{SDg} have the same premises above the line
and perform the same update on the global storage \(G\).
Moreover, no relay transaction is generated, so the mempool difference is empty,
and the local states \(\Psi_i\) remain unchanged. Hence there exists
\(C^{\mathsf{o}\prime}\) such that
\[
C^{\mathsf{o}} \to C^{\mathsf{o}\prime},
\qquad
\mathsf{Same}_g(C^{\mathsf{o}},C^{\mathsf{o}\prime}),
\qquad
\mathsf{MemDiff}^{\mathsf{o}}(C^{\mathsf{o}},C^{\mathsf{o}\prime},(\emptyset,\ldots,\emptyset)),
\]
and \(R_g(\gamma_G',C^{\mathsf{o}\prime})\) holds directly.

Conversely, suppose
\[
C^{\mathsf{o}} \to C^{\mathsf{o}\prime}
\]
is derived by rule \textsc{SDg}, together with
\[
\mathsf{Same}_g(C^{\mathsf{o}},C^{\mathsf{o}\prime}),
\qquad
\mathsf{MemDiff}^{\mathsf{o}}(C^{\mathsf{o}},C^{\mathsf{o}\prime},\Delta\Omega).
\]
Since \textsc{SDg} does not generate relay transactions, we have
\(\Delta\Omega=(\emptyset,\ldots,\emptyset)\).
Again, the premises of \textsc{SDg} and \textsc{SVDg} coincide exactly, so
there exists \(\gamma_G'\) such that
\[
\gamma_G \to_g^{\tau} \gamma_G'
\]
by \textsc{SVDg}, and \(R_g(\gamma_G',C^{\mathsf{o}\prime})\) follows
immediately.

\paragraph{Case \textsc{SASSIGNaa} (local).}
We show the local code-level bisimulation for rule \textsc{SASSIGNaa} in the
compositional semantics and rule \textsc{SAaa} in the original semantics.

Suppose \(R_l(i,\gamma_i,C^{\mathsf{o}})\) and
\[
\gamma_i \to_l^{?(w_1,\ldots,w_n,w_G)} \gamma_i'
\]
is derived by rule \textsc{SASSIGNaa}. The non-trivial part of the rule is the
evaluation of the right-hand-side expression \(exp\). By the induction hypothesis
applied to the corresponding evaluation derivation, the evaluation step in the
compositional semantics is simulated by the corresponding evaluation step in the
original semantics, with the same emitted relay information. In particular, the
mempool difference produced during evaluation is exactly the one required in the
statement-level step. After evaluation, the remaining premises of
\textsc{SASSIGNaa} and \textsc{SAaa} are identical: both rules update the same
\(@address\) state variable in the same local state component.
Therefore there exists \(C^{\mathsf{o}\prime}\) such that
\[
C^{\mathsf{o}} \to C^{\mathsf{o}\prime},
\qquad
\mathsf{Same}_l(i,C^{\mathsf{o}},C^{\mathsf{o}\prime}),
\qquad
\mathsf{MemDiff}^{\mathsf{o}}(C^{\mathsf{o}},C^{\mathsf{o}\prime},\Delta\Omega),
\]
where
\[
\Delta\Omega=(w_1\cup w_G,\ldots,w_n\cup w_G),
\]
and \(R_l(i,\gamma_i',C^{\mathsf{o}\prime})\) holds.

Conversely, suppose
\[
C^{\mathsf{o}} \to C^{\mathsf{o}\prime}
\]
is derived by rule \textsc{SAaa}, together with
\[
\mathsf{Same}_l(i,C^{\mathsf{o}},C^{\mathsf{o}\prime}),
\qquad
\mathsf{MemDiff}^{\mathsf{o}}(C^{\mathsf{o}},C^{\mathsf{o}\prime},\Delta\Omega).
\]
Again, the only non-trivial part is the evaluation of \(exp\). By the induction
hypothesis for the corresponding evaluation derivation in the original semantics,
there exists a matching evaluation step in the compositional semantics producing
the same relay information. The remaining premises of \textsc{SAaa} and
\textsc{SASSIGNaa} then coincide exactly, so there exists \(\gamma_i'\) such that
\[
\gamma_i \to_l^{?(w_1,\ldots,w_n,w_G)} \gamma_i'
\]
with
\[
w_j=\mathsf{LTx}(\Delta\Omega_j)\quad(1\le j\le n),
\qquad
w_G=\mathsf{GTx}(\Delta\Omega_j)
\]
for any \(j\), and \(R_l(i,\gamma_i',C^{\mathsf{o}\prime})\) follows.

\paragraph{Case \textsc{SASSIGNgg} (global).}
We show the global code-level bisimulation for rule \textsc{SASSIGNgg} in the
compositional semantics and rule \textsc{SAgg} in the original semantics.

Suppose \(R_g(\gamma_G,C^{\mathsf{o}})\) and
\[
\gamma_G \to_g^{?(w_1,\ldots,w_n,w_G)} \gamma_G'
\]
is derived by rule \textsc{SASSIGNgg}. As in the local assignment case, the
non-trivial part is the evaluation of the right-hand-side expression \(exp\).
By the induction hypothesis applied to the corresponding evaluation derivation,
the evaluation step in the compositional semantics is simulated by the matching
evaluation step in the original semantics, with the same emitted relay
information and hence the same mempool difference.

It remains to justify that the global correspondence relation is preserved for
the configurations occurring above the line of the rule. Here the crucial point
is that the global stack \(M_G\) in the compositional semantics corresponds to
the stacks \(M_i^{\mathsf{o}}\) of all engines in the original semantics.
Since we are in the global phase, these stacks are identical in the original
semantics, and \(Enc/Dec\) map them to the unique global stack \(M_G\) in the
compositional semantics. Therefore the configurations above the line still match
under \(R_g\).

After evaluation, the remaining premises of \textsc{SASSIGNgg} and
\textsc{SAgg} coincide: both rules update the same global state variable in the
same way. Hence there exists \(C^{\mathsf{o}\prime}\) such that
\[
C^{\mathsf{o}} \to C^{\mathsf{o}\prime},
\qquad
\mathsf{Same}_g(C^{\mathsf{o}},C^{\mathsf{o}\prime}),
\qquad
\mathsf{MemDiff}^{\mathsf{o}}(C^{\mathsf{o}},C^{\mathsf{o}\prime},\Delta\Omega),
\]
where
\[
\Delta\Omega=(w_1\cup w_G,\ldots,w_n\cup w_G),
\]
and \(R_g(\gamma_G',C^{\mathsf{o}\prime})\) holds.

Conversely, suppose
\[
C^{\mathsf{o}} \to C^{\mathsf{o}\prime}
\]
is derived by rule \textsc{SAgg}, together with
\[
\mathsf{Same}_g(C^{\mathsf{o}},C^{\mathsf{o}\prime}),
\qquad
\mathsf{MemDiff}^{\mathsf{o}}(C^{\mathsf{o}},C^{\mathsf{o}\prime},\Delta\Omega).
\]
Again, by the induction hypothesis for the corresponding evaluation derivation,
there exists a matching evaluation step in the compositional semantics
producing the same relay information. The correspondence above the line is
preserved because the common stacks \(M_i^{\mathsf{o}}\) in the original
semantics are represented by the unique stack \(M_G\) in the compositional
semantics. The remaining premises of \textsc{SAgg} and \textsc{SASSIGNgg}
then coincide exactly, so there exists \(\gamma_G'\) such that
\[
\gamma_G \to_g^{?(w_1,\ldots,w_n,w_G)} \gamma_G'
\]
with
\[
w_i=\mathsf{LTx}(\Delta\Omega_i)\quad(1\le i\le n),
\qquad
w_G=\mathsf{GTx}(\Delta\Omega_i)
\]
for any \(i\), and \(R_g(\gamma_G',C^{\mathsf{o}\prime})\) follows.

\paragraph{Case \textsc{IFUNaa} (local).}
We show the local code-level bisimulation for rule \textsc{IFUNaa} in the
compositional semantics and rule \textsc{IFaa} in the original semantics.

Suppose \(R_l(i,\gamma_i,C^{\mathsf{o}})\) and
\[
\gamma_i \to_l^{?(w_1,\ldots,w_n,w_G)} \gamma_i'
\]
is derived by rule \textsc{IFUNaa}. This is the standard function-call case.
Both \textsc{IFUNaa} and \textsc{IFaa} first create a new stack frame, assign
the same scope information and return marker, initialize the formal parameters
in the same way, execute the same function body, and finally remove the top
frame.

To prove the result, we first check that the configurations occurring above the
line still satisfy the local correspondence relation \(R_l\). This holds because
the newly added stack frame and the parameter initialization are constructed in
the same way in both semantics, while the surrounding components remain
related exactly as required by \(R_l\).

We then apply the induction hypothesis to the execution of the function body,
namely to the step above the line
\[
(\sigma_i^{(1)},\, T_1\,id_1=exp_1;\ \cdots\ ;\ T_t\,id_t=exp_t;\ Block)
\]
in the compositional semantics and its counterpart in the original semantics.
The induction hypothesis yields the matching step in the other semantics,
together with the required mempool difference. Since the push/pop operations
and the parameter declarations coincide on both sides, the outer function-call
rules \textsc{IFUNaa} and \textsc{IFaa} reconstruct exactly corresponding final
configurations. Therefore there exists \(C^{\mathsf{o}\prime}\) such that
\[
C^{\mathsf{o}} \to C^{\mathsf{o}\prime},
\qquad
\mathsf{Same}_l(i,C^{\mathsf{o}},C^{\mathsf{o}\prime}),
\qquad
\mathsf{MemDiff}^{\mathsf{o}}(C^{\mathsf{o}},C^{\mathsf{o}\prime},\Delta\Omega),
\]
where
\[
\Delta\Omega=(w_1\cup w_G,\ldots,w_n\cup w_G),
\]
and \(R_l(i,\gamma_i',C^{\mathsf{o}\prime})\) holds.

Conversely, suppose
\[
C^{\mathsf{o}} \to C^{\mathsf{o}\prime}
\]
is derived by rule \textsc{IFaa}, together with
\[
\mathsf{Same}_l(i,C^{\mathsf{o}},C^{\mathsf{o}\prime}),
\qquad
\mathsf{MemDiff}^{\mathsf{o}}(C^{\mathsf{o}},C^{\mathsf{o}\prime},\Delta\Omega).
\]
We again first observe that the configurations above the line satisfy the
corresponding \(R_l\) relation. Applying the induction hypothesis to the body
execution yields a matching local step in the compositional semantics with the
same relay information. Since the surrounding push/pop and parameter-binding
operations coincide, we obtain a component \(\gamma_i'\) such that
\[
\gamma_i \to_l^{?(w_1,\ldots,w_n,w_G)} \gamma_i'
\]
with
\[
w_j=\mathsf{LTx}(\Delta\Omega_j)\quad(1\le j\le n),
\qquad
w_G=\mathsf{GTx}(\Delta\Omega_j)
\]
for any \(j\), and \(R_l(i,\gamma_i',C^{\mathsf{o}\prime})\) follows.

\paragraph{Case \textsc{IFUNgg} (global).}
We show the global code-level bisimulation for rule \textsc{IFUNgg} in the
compositional semantics and rule \textsc{IFgg} in the original semantics.

Suppose \(R_g(\gamma_G,C^{\mathsf{o}})\) and
\[
\gamma_G \to_g^{?(w_1,\ldots,w_n,w_G)} \gamma_G'
\]
is derived by rule \textsc{IFUNgg}. This is the most involved case among all
representative rules, because the proof must carefully maintain the global
correspondence relation \(R_g\) across the intermediate configurations occurring
above the line.

Both \textsc{IFUNgg} and \textsc{IFgg} perform the same sequence of operations:
they create a new stack frame, assign the same scope and return marker,
initialize the formal parameters in the same way, execute the same function
body, and finally remove the top frame. The essential difficulty is that the
compositional semantics uses the single global stack \(M_G\), whereas the
original semantics uses one stack \(M_i^{\mathsf{o}}\) in each engine.

We first verify that the configurations above the line still satisfy the
required \(R_g\)-correspondence.

For the global state \(G\), this is immediate from the rule itself.
Before the body starts executing, the state component \(G\) is not modified by
the stack-extension step, so the \(G\) occurring in \(\sigma_G^{(1)}\) is still
the same as the global state in the corresponding configuration of the original
semantics. By the induction hypothesis applied to the execution of the function
body, the global state in \(\sigma_G^{(2)}\) corresponds to the updated global
state \(G'\) in the original semantics. The final pop step does not modify the
global storage either, so the \(G\) occurring in \(\sigma_G^{(3)}\) is still the
same updated \(G'\).

For the stack component, we again compare the intermediate configurations
above the line. By replacing the occurrences of \(M_G\) in the premises of
\textsc{IFUNgg} with the corresponding stacks \(M_i^{\mathsf{o}}\) of the
original semantics, we see that the stack-extension step constructs exactly the
same new frame on both sides. Hence the stack \(M_G^{(1)}\) in
\(\sigma_G^{(1)}\) corresponds to the common stack \(M_i^{\mathsf{o}(1)}\) in
the original semantics. Next, by the induction hypothesis for the execution of
the function body, the stack \(M_G^{(2)}\) in \(\sigma_G^{(2)}\) corresponds to
the common stack \(M_i^{\mathsf{o}(2)}\) in the original semantics. Finally,
the pop step removes the top frame in the same way on both sides, so
\(M_G^{(3)}\) corresponds to \(M_i^{\mathsf{o}(3)}\). Since the conclusion of
\textsc{IFgg} gives \(\sigma_i^{\mathsf{o}}=\sigma_i^{\mathsf{o}(3)}\), it
follows that the final stack component is again aligned. Thus the intermediate
and final configurations all satisfy the required \(R_g\)-correspondence.

Once these intermediate \(R_g\)-relations are established, the rest of the proof
is standard. We apply the induction hypothesis to the execution of the function
body above the line. This yields a matching step in the original semantics with
the same relay information and the corresponding mempool difference. Since the
surrounding push/pop operations and parameter initialization coincide on both
sides, we obtain a configuration \(C^{\mathsf{o}\prime}\) such that
\[
C^{\mathsf{o}} \to C^{\mathsf{o}\prime},
\qquad
\mathsf{Same}_g(C^{\mathsf{o}},C^{\mathsf{o}\prime}),
\qquad
\mathsf{MemDiff}^{\mathsf{o}}(C^{\mathsf{o}},C^{\mathsf{o}\prime},\Delta\Omega),
\]
where
\[
\Delta\Omega=(w_1\cup w_G,\ldots,w_n\cup w_G),
\]
and \(R_g(\gamma_G',C^{\mathsf{o}\prime})\) holds.

Conversely, suppose
\[
C^{\mathsf{o}} \to C^{\mathsf{o}\prime}
\]
is derived by rule \textsc{IFgg}, together with
\[
\mathsf{Same}_g(C^{\mathsf{o}},C^{\mathsf{o}\prime}),
\qquad
\mathsf{MemDiff}^{\mathsf{o}}(C^{\mathsf{o}},C^{\mathsf{o}\prime},\Delta\Omega).
\]
The proof proceeds symmetrically. We first verify that the configurations above
the line satisfy the required \(R_g\)-correspondence: the global storage remains
aligned before and after the stack operations, and the common engine stacks
\(M_i^{\mathsf{o}(1)}, M_i^{\mathsf{o}(2)}, M_i^{\mathsf{o}(3)}\) in the
original semantics correspond to the single stacks
\(M_G^{(1)}, M_G^{(2)}, M_G^{(3)}\) in the compositional semantics. We then
apply the induction hypothesis to the body execution, obtaining a matching
global step in the compositional semantics with the same relay information.
Since the surrounding push/pop and parameter-binding operations coincide, there
exists \(\gamma_G'\) such that
\[
\gamma_G \to_g^{?(w_1,\ldots,w_n,w_G)} \gamma_G'
\]
with
\[
w_i=\mathsf{LTx}(\Delta\Omega_i)\quad(1\le i\le n),
\qquad
w_G=\mathsf{GTx}(\Delta\Omega_i)
\]
for any \(i\), and \(R_g(\gamma_G',C^{\mathsf{o}\prime})\) follows.

\paragraph{Case \textsc{RELAYal} (local).}
We show the local code-level bisimulation for rule \textsc{RELAYal} in the
compositional semantics and rule \textsc{RELa} in the original semantics.

Suppose \(R_l(i,\gamma_i,C^{\mathsf{o}})\) and
\[
\gamma_i \to_l^{?(w_1,\ldots,w_n,w_G)} \gamma_i'
\]
is derived by rule \textsc{RELAYal}. The two rules perform the same local
state change, namely no change to the active state component, and differ only
in how the emitted relay information is represented. In the compositional
semantics, the emitted label is
\[
w_r=\{(address,j,idf(exp_1,\ldots,exp_t))\}, \qquad
w_\iota=\emptyset \text{ for } \iota\neq r, \qquad
w_G=\emptyset,
\]
while in the original semantics the same relay transaction is added to the
mempool of engine \(r\). Hence the required mempool difference is immediate:
\[
\Delta\Omega_r = w_r, \qquad
\Delta\Omega_\iota = \emptyset \text{ for } \iota\neq r.
\]
Therefore there exists \(C^{\mathsf{o}\prime}\) such that
\[
C^{\mathsf{o}} \to C^{\mathsf{o}\prime},
\qquad
\mathsf{Same}_l(i,C^{\mathsf{o}},C^{\mathsf{o}\prime}),
\qquad
\mathsf{MemDiff}^{\mathsf{o}}(C^{\mathsf{o}},C^{\mathsf{o}\prime},\Delta\Omega),
\]
and \(R_l(i,\gamma_i',C^{\mathsf{o}\prime})\) holds directly.

Conversely, suppose
\[
C^{\mathsf{o}} \to C^{\mathsf{o}\prime}
\]
is derived by rule \textsc{RELa}, together with
\[
\mathsf{Same}_l(i,C^{\mathsf{o}},C^{\mathsf{o}\prime}),
\qquad
\mathsf{MemDiff}^{\mathsf{o}}(C^{\mathsf{o}},C^{\mathsf{o}\prime},\Delta\Omega).
\]
By inspection of the rule, the only effect is the insertion of one address relay
transaction into the appropriate engine mempool. Thus
\[
w_r=\mathsf{LTx}(\Delta\Omega_r), \qquad
w_\iota=\emptyset \text{ for } \iota\neq r, \qquad
w_G=\emptyset,
\]
and there exists \(\gamma_i'\) such that
\[
\gamma_i \to_l^{?(w_1,\ldots,w_n,w_G)} \gamma_i'
\]
by \textsc{RELAYal}. The relation \(R_l(i,\gamma_i',C^{\mathsf{o}\prime})\)
then follows immediately.

\paragraph{Case \textsc{RELAYgl1} (local).}
We show the local code-level bisimulation for rule \textsc{RELAYgl1} in the
compositional semantics and rule \textsc{RELg1} in the original semantics.

Suppose \(R_l(i,\gamma_i,C^{\mathsf{o}})\) and
\[
\gamma_i \to_l^{?(w_1,\ldots,w_n,w_G)} \gamma_i'
\]
is derived by rule \textsc{RELAYgl1}. Again, the step does not change the
active state component; the only observable effect is the generated relay
information. In the compositional semantics, the label contains
\[
w_G=\{(global,idf(exp_1,\ldots,exp_t))\}, \qquad
w_\iota=\emptyset \text{ for all local mempool components } \iota.
\]
In the original semantics, the corresponding global relay transaction is added
to every engine mempool. Therefore the required mempool difference is exactly
\[
\Delta\Omega_\iota = w_G \qquad (1\le \iota \le n).
\]
Hence there exists \(C^{\mathsf{o}\prime}\) such that
\[
C^{\mathsf{o}} \to C^{\mathsf{o}\prime},
\qquad
\mathsf{Same}_l(i,C^{\mathsf{o}},C^{\mathsf{o}\prime}),
\qquad
\mathsf{MemDiff}^{\mathsf{o}}(C^{\mathsf{o}},C^{\mathsf{o}\prime},\Delta\Omega),
\]
and \(R_l(i,\gamma_i',C^{\mathsf{o}\prime})\) holds directly.

Conversely, suppose
\[
C^{\mathsf{o}} \to C^{\mathsf{o}\prime}
\]
is derived by rule \textsc{RELg1}, together with
\[
\mathsf{Same}_l(i,C^{\mathsf{o}},C^{\mathsf{o}\prime}),
\qquad
\mathsf{MemDiff}^{\mathsf{o}}(C^{\mathsf{o}},C^{\mathsf{o}\prime},\Delta\Omega).
\]
By inspection of the rule, the only effect is the insertion of the same global
relay transaction into all engine mempools. Thus
\[
w_\iota=\mathsf{LTx}(\Delta\Omega_\iota)=\emptyset \qquad (1\le \iota \le n),
\]
and
\[
w_G=\mathsf{GTx}(\Delta\Omega_\iota)
\]
for any \(\iota\), which is well defined by the agreement of global transactions
in mempools. Hence there exists \(\gamma_i'\) such that
\[
\gamma_i \to_l^{?(w_1,\ldots,w_n,w_G)} \gamma_i'
\]
by \textsc{RELAYgl1}, and \(R_l(i,\gamma_i',C^{\mathsf{o}\prime})\) follows
immediately.

\paragraph{Remaining cases.}
The remaining code-level rules can be handled similarly.
Rules in Class~B-E that are not shown explicitly are entirely analogous to the representative assignment and function-call cases presented above.
Rules in Class~A are purely structural and follow directly from the induction hypothesis on the corresponding sub-derivations.

\paragraph{Conclusion of the proof.}
Since every code-level rule of the compositional semantics has a corresponding rule in the original semantics and vice versa, and each case preserves the relations \(R_l\) or \(R_g\) together with the required mempool differences, the statements of Theorems~\ref{lem:local-code-bisim} and \ref{lem:global-code-bisim} follow by induction on the derivation of the transition step.

\end{proof}

\begin{proof}[Proof of Theorem~\ref{thm:transaction-bisim}]
The proof reduces transaction-level execution to the execution of the
corresponding component after removing the transaction wrapper.
We prove the two directions separately.

\smallskip
\noindent
\textbf{(i) From the original semantics to the compositional semantics.}
Assume \(R_T(C^{\mathsf{o}},C)\) and \( C^{\mathsf{o}} \to C^{\mathsf{o}\prime} \) by a transaction rule of the original semantics.

\smallskip
\noindent
\emph{Case 1.1: the step is derived by \textsc{IFUNtg}.}
Write the original configuration \(C^{\mathsf{o}}\) explicitly and compute the corresponding compositional configuration \(C\) using \(Enc_T\).
Since the transaction is global, to obtain the desired system step in the
compositional semantics it suffices, by rule \textsc{G}, to prove that the
global component \(\gamma_G\) performs the corresponding code-level step.

Inspecting the premises of the rule, all information above the line coincides in the two semantics after removing the outer transaction wrapper. Therefore the problem reduces directly to the corresponding code-level step of the global component. This follows immediately from the global code-level bisimulation lemma.

\smallskip
\noindent
\emph{Case 1.2: the step is derived by \textsc{IFUNta} or \textsc{IFUNts}.}
Again write the configuration explicitly and compute the corresponding
configuration of the compositional semantics using \(Enc_T\).
Since the transaction is local, by rule \textsc{L} it suffices to show that the corresponding local component \(\gamma_i\) performs the matching code-level step.

After removing the outer transaction wrapper, the premises above the line of the rules coincide in the two semantics. Hence the required local component step follows directly from the local code-level bisimulation lemma.

\smallskip
\noindent
\textbf{(ii) From the compositional semantics to the original semantics.}
Assume \(R_T(C^{\mathsf{o}},C)\) and
\( C \sysstep C'\).

\smallskip
\noindent
\emph{Case 2.1: the step is a system-level global step.}
Write the compositional configuration \(C\) explicitly and compute the
corresponding original configuration \(C^{\mathsf{o}}\) using \(Dec_T\).
By inversion on the system-level rule, the step must have been derived by rule \textsc{G}. Therefore the global component \(\gamma_G\) performs a code-level step.

Inspecting the premises above the line shows that all information coincides
with the corresponding configuration of the original semantics once the
transaction wrapper is removed. The required step in the original semantics
then follows immediately from the global code-level bisimulation lemma.

\smallskip
\noindent
\emph{Case 2.2: the step is a system-level local step.}
Write the configuration explicitly and compute the corresponding original
configuration using \(Dec_T\). By inversion on the system rule, the step must have been derived by rule \textsc{L}, so some local component \(\gamma_i\) performs a code-level step.

After removing the transaction wrapper, the premises above the line coincide with those of the original semantics. The required step of the original semantics therefore follows directly from the local code-level bisimulation lemma.

\smallskip
In all cases the resulting configurations again satisfy the relation \(R_T\).
\end{proof}

\end{document}